\documentclass[11pt,a4paper]{article}

\usepackage[utf8]{inputenc} 
\usepackage[T1]{fontenc}    %
\usepackage[english]{babel} 
\usepackage[final]{microtype} 

\usepackage{amssymb,amsmath,mathtools} 
\usepackage{amsthm} 

\usepackage{xcolor}
\usepackage{bbm}
\usepackage{braket}
\usepackage{appendix}
\usepackage{mathtools}

\usepackage[hmargin=0.12\paperwidth,vmargin=0.16\paperwidth,bindingoffset=0cm]%
{geometry} 

\numberwithin{equation}{section} 

\theoremstyle{plain}
  
  \newtheorem{prop}{Proposition}
  
  \newtheorem{cor}{Corollary}
\theoremstyle{definition}
  
  \newtheorem{remark}{Remark}
  
  \numberwithin{prop}{section}
   \numberwithin{cor}{section}
   \numberwithin{remark}{section}

\DeclareMathOperator*{\pf}{Pf} 

\title{\Large\bfseries Finite size corrections in the bulk for circular $\beta$ ensembles}%
   
\author{Peter J. Forrester and Bo-Jian Shen}
\date{}

\begin{document}

\maketitle

School of Mathematics and Statistics,  The University of Melbourne,
Victoria 3010, Australia. \: \: Email: {\tt pjforr@unimelb.edu.au}; {\tt bojian.shen@unimelb.edu.au}\\

\bigskip

\begin{abstract}
\noindent
The circular $\beta$ ensemble for $\beta =1,2$ and 4 corresponds to circular orthogonal, unitary and symplectic ensemble respectively as introduced by Dyson. The statistical state of the eigenvalues is then a determinantal point process ($\beta = 2$) and Pfaffian point process ($\beta = 1,4$). The explicit functional forms of the correlation kernels then imply that the general $n$-point correlation functions exhibit an asymptotic expansion in $1/N^2$, which moreover can be lifted to an asymptotic in $1/N^2$ for the spacing distributions and their generating function. 
We use $\sigma$-Painlev\'e characterisations to show that the functional form of the first correction is related to the leading term via a second derivative. 
In the case $\beta = 2$ this finding has immediate consequence in interpreting the empirical Riemann zeros spacing distribution at large height, and that of their thinning.
Explicit functional forms are used to show that the spectral form factors for $\beta =1,2$ and 4 also admit an asymptotic expansion in $1/N^2$. Differential relations are identified expressing the first and second correction in terms of the limiting functional form, and evidence is presented that they hold for general $\beta$. For even $\beta$ it is proved that the two-point correlation function permits an asymptotic expansion in $1/N^2$, and moreover that the 
leading correction relates to the limiting functional form via a second derivative.

\end{abstract}

\vspace{3em}

\section{Introduction}

\subsection{Bulk scaling for the CUE}
Dyson's circular unitary ensemble (CUE) \cite{Dy62} is identical to the set of $N \times N$ matrices from the classical unitary group chosen with Haar measure. From the latter viewpoint, which has its origins in the work of Hurwitz (see the review \cite{DF17}), it is a classical result of Weyl \cite{We39} that the corresponding eigenvalue probability density function (PDF) is given by
\begin{equation}\label{1.1}
p_{2,N}(\theta_1,\dots,\theta_N) := {1 \over (2 \pi)^N N!}
\prod_{1 \le j < k \le N} | e^{i \theta_k} - e^{i \theta_j}|^2,
\quad 0 \le \theta_l < 2 \pi \: \: (l=1,\dots,N).
\end{equation}

From this starting point, Dyson \cite{Dy62a} deduced that the general $n$-point correlation function
\begin{equation}\label{1.2}
\rho_{(n)}^{\rm CUE}(\theta_1,\dots,\theta_k) :=
N (N-1) \cdots (N-n+1) \int_0^{2 \pi} d \theta_{n+1} \cdots \int_0^{2 \pi} d \theta_N \, p_{2,N}(\theta_1,\dots,\theta_N)
\end{equation}
has the determinantal form
\begin{equation}\label{1.3}
\rho_{(k)}^{\rm CUE}(\theta_1,\dots,\theta_n)  =
\det [ K_N^{\rm CUE} (\theta_j, \theta_k) ]_{j,k=1,\dots,n},
\end{equation}
where
\begin{equation}\label{1.4}
K_N^{\rm CUE} (\theta, \phi) = {1 \over 2 \pi} {\sin (N(\theta - \phi)/2) \over \sin((\theta - \phi)/2)}.
\end{equation}
Dyson also considered the bulk scaled $N \to \infty$ limit, in which the angles are scaled $\theta_l = 2 \pi X_l/N$ so that in the variables $\{X_l\}$ the mean eigenvalue density is unity, with the result
\begin{equation}\label{2.1}
\rho_{(n)}^{\rm bulk}(X_1,\dots,X_n) :=
\lim_{N \to \infty} \Big ( {2 \pi \over N} \Big )^n
\rho_{(n)}^{\rm CUE}\bigg ( {2 \pi X_1 \over N},\dots,
 {2 \pi X_n \over N} \bigg ) = \det [ K_\infty (X_j, X_k) ]_{j,k=1,\dots,n},
\end{equation} 
where
\begin{equation}\label{2.1a}
K_\infty(X,Y) := \lim_{N \to \infty} {2 \pi \over N}
K_N^{\rm CUE}  \bigg ( {2 \pi X \over N}, {2 \pi Y \over N} \bigg ) =
{\sin \pi (X - Y) \over \pi (X - Y) }.
\end{equation}

Our interest in the present paper relates to the structure of the large $N$ expansion of the bulk scaled $n$-point correlation, which is the function of $\{X_l\}$ and $N$ in the second expression in (\ref{2.1}). Beyond the CUE, we take up this question for the class of circular $\beta$-ensembles specified by the 
eigenvalue  PDF
\begin{equation}\label{1.1A}
p_{\beta,N}(\theta_1,\dots,\theta_N) := {1 \over (2 \pi)^N }
{1 \over C_{\beta,N}}
\prod_{1 \le j < k \le N} | e^{i \theta_k} - e^{i \theta_j}|^\beta,
\quad 0 \le \theta_l < 2 \pi \: \: (l=1,\dots,N),
\end{equation}
where
\begin{equation}\label{1.1B}
C_{\beta,N} = {\Gamma(\beta N + 1) \over (\Gamma(\beta + 1) )^N}.
\end{equation}
Note that the CUE is the case $\beta = 2$ of (\ref{1.1A}). For a construction of Hessenberg unitary matrices which realise (\ref{1.1A}) for general $\beta > 0$, and moreover show that the bulk scaled limit is well defined, see \cite{KN04}.

In the case of the CUE, due to (\ref{1.3}), the question of the structure of the large $N$ expansion of the bulk scaled $k$-point correlation function reduces to the simple question of expanding for large $N$
\begin{equation}\label{3.1}
{1 \over N} {\sin \pi (X - Y) \over \sin ( \pi (X - Y)/N)}.
\end{equation}
A significant feature is that in this functional form $N$ need not be restricted to positive integer values, but rather may be regarded as a continuous parameter. Adopting this viewpoint, 
one sees that it is an even function of $N$ (as well as in $(X-Y)$), possessing an expansion in powers of $1/N^2$
\begin{equation}\label{3.2}
{\sin \pi (X - Y) \over \pi (X - Y) } + {1 \over N^2} 
{\pi \over 6} (X - Y) \sin \pi (X - Y) + {\rm O} \bigg (
{1 \over N^4} \bigg ).
\end{equation}
In particular, for the bulk scaled two-point correlation function, this implies that for large $N$
\begin{equation}\label{3.3}
\Big ( {2 \pi \over N} \Big )^2
\rho_{(2)}^{\rm CUE}\bigg ( {2 \pi X \over N},
 {2 \pi Y \over N} \bigg ) =
 1 - \bigg ( {\sin \pi (X - Y) \over \pi (X - Y) } \bigg )^2 -
 {1 \over 3 N^2}  \sin^2 \pi (X - Y) + {\rm O}
 \bigg (
{1 \over N^4} \bigg ),
\end{equation}
a result which has application in interpreting data for the two-point correlation function of the Riemann zeros at large height
\cite{BBLM06,FM15,BFM17}. More on the consequences of our results in relation to structure seen in particular empirical distributions for
the Riemann zeros at large height is to come latter in this article.
At present the main point to be emphasised is that the asymptotic expansion of the general $k$-point correlation function, to all orders in $1/N$, contains only powers of 
$1/N^2$,
\begin{equation}\label{4.1}
\Big ( {2 \pi \over N} \Big )^n
\rho_{(n)}^{\rm CUE}\bigg ( {2 \pi X_1 \over N},\dots,
 {2 \pi X_n \over N} \bigg ) \sim
 \rho_{(n),0,\beta = 2}^{\rm bulk}(X_1,\dots,X_n) +
 \sum_{l=1}^\infty {1 \over N^{2l}} \rho_{(n),l,\beta = 2}^{\rm bulk}(X_1,\dots,X_n)
 \end{equation}
 for some $\{ \rho_{(n),l,\beta = 2}^{\rm bulk}(X_1,\dots,X_n) \}$.

 Establishing such a result has consequences for the large $N$ expansion of the probability density function of the distribution function for the spacing between a given eigenvalue and its $l$-th neighbour to the right. These distribution functions are fully determined by the probabilities $\{ E_N(k;(0,\phi)) \}_{k=0}^N$ for the interval $(0,\phi)$ containing exactly $k$ eigenvalues, which in turn have for their generating function the expansion
 \begin{equation}\label{4.2}
 \mathcal E_N^{(\cdot)}((0,\phi);\xi)
 := \sum_{k=0}^N (1 - \xi)^k
  E_N^{(\cdot)}(k;(0,\phi))
 = 1 +
 \sum_{k=1}^N {(-\xi)^k \over k!} \int_0^\phi d \theta_1 \cdots
 \int_0^\phi d \theta_k \,  \rho_{(k)}^{(\cdot)}(\theta_1,\dots,\theta_k) ,
  \end{equation}
  valid for any one-dimensional point process $(\cdot)$ on the circle
 (see e.g.~\cite[\S 8.1]{Fo10} in relation to the above theory).
 Thus, specialising to the case of the CUE, it follows by substituting (\ref{4.1}) in (\ref{4.2}) that for large $N$
  \begin{equation}\label{4.3}
 \mathcal E_N^{\rm CUE}((0,2 \pi s / N);\xi) \sim
\mathcal E^{\rm bulk}_{0,\beta = 2}((0,s);\xi) + \sum_{l=1}^\infty {1 \over N^{2l} } \mathcal E^{\rm bulk}_{l,\beta = 2}((0,s);\xi),
  \end{equation}
  for some $\{\mathcal E^{\rm bulk}_{l,\beta = 2}((0,s);\xi)  \}$ (convergence of the implied infinite sums at each order which form the $\mathcal E^{\rm bulk}_{l,\beta = 2}((0,s);\xi)$ can readily be established;
  see \cite[Lemma 2.1]{BFM17}).
  In fact one has available the explicit 
  small-$s$ expansion
  (\cite{TW94c}, \cite[Eq.~(8.79), extended to two further terms according to the methodology therein]{Fo10})
  \small \begin{align}\label{1.15}
  & \mathcal E_N^{\rm CUE}((0,2 \pi s / N);\xi) = 1 - \xi s +
  {(1 - 1/N^2) \xi^2 \pi^2 s^4 \over 36} -
  {(1-1/N^2)(2 - 3/N^2) \over 1350} \xi^2 \pi^4 s^6 \nonumber \\
 & + {(1-1/N^2)(1 - 2/N^2)(3 - 5/N^2) \over 52920} \xi^2 \pi^6 s^8 -{(1-4/N^2)(1 - 1/N^2)^2 \over 291600} \xi^3 \pi^6 s^9  \nonumber\\ 
 & - {(1-1/N^2)(2 - 3/N^2)(1 - 5/N^2+7/N^4) \over 1275750} \xi^2 \pi^8 s^{10} +{(1-1/N^2) (1 - 4/N^2)^2(6 - 19/N^2) \over 29767500} \xi^3 \pi^8 s^{11}  \nonumber\\
 & +{\rm O}(s^{12}),
  \end{align}
  \normalsize
  which is immediately observed to permit an expansion in $1/N^2$ to each order in $s$ exhibited. In keeping with the text below (\ref{3.3}), one remarks that computable functional forms of the term at order $1/N^2$ for the distribution of the spacing between eigenvalues, and this distribution after a thinning of the eigenvalue sequence (since \cite{BP04} it is known that the latter is
  controlled by the parameter $\xi$), both of which relate to $\mathcal E^{\rm bulk}_{l,\beta = 2}((0,s);\xi) |_{l=1}$,
   have been used in \cite{FM15,BFM17} for purposes of interpreting empirically determined spacing distributions for the Riemann zeros at large height.
   A major finding of the present work in the regards --- see Proposition \ref{P2.1} --- is that correction at order $1/N^2$ is simply related to the leading distribution by a derivative operation.

  \subsection{Relationship to earlier work}
  Prominence to the structure of the large $N$ asymptotic expansion of quantities in random matrix theory comes from several sources. The first, and probably the best known, relates to the $1/N$ expansion associated with the so-called loop equations \cite{Mi04}. In this formalism the connected correlators for the product of linear statistics 
  \begin{equation}\label{LS}
  \prod_{l=1}^s \sum_{j=1}^N {1 \over y_l - \lambda_j},
  \end{equation}
  each of which can be considered as the generating function for products of $s$ monomials, are by an hypothesis on the existence of the $1/N$ expansion
  (this hypotheses has subsequently been proved in a number of prominent settings \cite{BG11})
  in the so-called global scaling limit fully determined by a triangular system of equations. As a concrete example, consider following \cite{BMS11,MMPS12,WF14}
  a loop equation analysis applied to the Gaussian $\beta$ ensemble characterised by the eigenvalue PDF proportional to 
  \begin{equation}\label{G1}
  \prod_{l=1}^N e^{-\beta N \lambda_l^2}
  \prod_{1 \le j < k \le N} | \lambda_k - \lambda_j|^\beta
   \end{equation}
   (specifically here the $N$-dependent factor in the Gaussian implies the limiting density is supported on a compact interval, $(-1,1)$ to be precise, which is the requirement of global scaling). In relation to the averaged monomials $\tilde{m}_k^{\rm G}(\beta,N) :=
   N^{-k/2-1} \langle \sum_{l=1}^N \lambda_l^k \rangle$, where the average is with respect to (\ref{G1}), the $1/N$ expansion in fact terminates at order $N^{-k/2}$. As explicit examples, taking into account that by symmetry all the odd moments vanish, for the second and fourth moments one can calculate that
   (see e.g.~\cite{DE06})
   \begin{align}\label{D1x}
   \tilde{m}_2^{\rm G}(\beta,N) & = 1 + N^{-1}(-1 + \tau^{-1}), \nonumber \\
  \tilde{m}_4^{\rm G}(\beta,N) & = 2 + 5 N^{-1}(-1 + \tau^{-1}) +
  N^{-2} (3 - 5 \tau^{-1} + 3 \tau^{-2}),
  \end{align}
  where $\tau:=\beta/2$. It is furthermore true that the moments satisfy the duality relation \cite{DE06}
   \begin{equation}\label{D1}
  \tilde{m}_k^{\rm G}(\beta,N)  = \tilde{m}_k^{\rm G}(4/\beta,
  -\beta N/2).
  \end{equation}
  (As in (\ref{3.1}), although now for a different reason, namely that the quantities in question are polynomials in all parameters, both sides of (\ref{D1}) have meaning for continuous $N$.)
  In particular, the duality relation implies that for $\beta = 2$ the moments are even in $N$ and so have an expansion in terms of $1/N^2$; for $\beta \ne 2$ the explicit results of (\ref{D1x}) exhibit a $1/N$ expansion.

  At the edge of the spectrum the Gaussian $\beta$-ensemble permits a soft edge scaling
  \cite{Fo93a}
   \begin{equation}\label{G1a}
  \lambda_l \mapsto 1 + x_l/2 N^{2/3}
   \end{equation}
  to a
  well defined statistical state with the eigenvalues spaced at order unity apart. Recent work of Bornemann \cite{Bo24,Bo25a,Bo25b} has for $\beta = 1,2$ and 4 established an asymptotic expansion of the density and the spacing distribution in powers of $(N')^{-2/3}$, $N' := N + {2 - \beta \over 2 \beta}$, with moreover functional forms in the corrections relating to the limiting functional form by certain derivative operations, multiplied by a polynomial\footnote{\label{FN1}Some assumptions are required. The mildest relates to the expansion being in powers of $(N')^{-2/3}$ to all orders for the densities, which in \cite{Bo25a} is shown to hold up to the first $m_* = 10$ terms from the corresponding soft edge expansion of the Christoffel-Darboux kernel in \cite{Bo24}. As commented in \cite{Bo25a} it is expected that a Riemann-Hilbert analysis analogous to that carried out in \cite{YZ23} can validate that the soft edge scaling asymptotic expansion of the Christoffel-Darboux kernel to all orders is in powers of $N^{-2/3}$.}
  . For example, in the case $\beta=2$, when (\ref{G1}) corresponds to the Gaussian unitary ensemble GUE${}^*$ --- here the asterisk indicates the use of global scaling with the limiting density supported on $(-1,1)$ --- one has 
  \cite[Th.~2.1 in the case $\xi = 1$]{Bo24} (extended to general $\xi$ in  \cite{Bo25b})
   \begin{multline}\label{D2}
   {1 \over 2 N^{2/3}}
  E_N^{{\rm GUE}{}^*}((1 + t/2 N^{2/3}, \infty);\xi) =  
   E_{\beta = 2}^{ \rm soft}((t, \infty);\xi)  \\
  + N^{-2/3} \Big (
   {t^2 \over 2} {d \over d t} E_{\beta = 2}^{ {\rm soft}}((t, \infty);\xi) -
   {3 \over 10} {d^2 \over d t^2} E_{\beta =2}^{{\rm soft}}((t, \infty);\xi) \Big ) +
   {\rm O} (N^{-4/3}).
   \end{multline}
   In \cite{FT19a}, in relation to the spectral density for the Gaussian $\beta$ ensemble with $\beta$ even, it was established that the leading correction to the limiting density occurs at order $(N')^{-2/3}$. Strictly speaking, the result was stated in terms of an $N$ dependent shift in $t$, which is equivalent to the stated $\beta$ dependent shift in the meaning of $N'$ up to this order; on this point  in relation to the cases $\beta =1$ and 2, see also 
   \cite{JM12}. One remarks too that the expansion (\ref{D2}), but with a weaker error bound, and a more complicated (but still computable) functional form of the $N^{-2/3}$ term was given in \cite{FT18}. 

   In the case of the Laguerre $\beta$-ensemble, specified by the eigenvalue PDF proportional to
   \begin{equation}\label{D3} 
   \prod_{l=1}^N \lambda_l^a e^{-\beta \lambda_l/2}
   \mathbbm 1_{\lambda_l > 0}
   \prod_{1 \le j < k \le N} | \lambda_k - \lambda_j |^\beta,
    \end{equation}
  the works of Bornemann \cite{Bo24,Bo25a} establish an analogous result at the soft edge (scaled neighbourhood of the largest eigenvalue) for $\beta = 1,2$ and 4, although
  the parameter playing the role of $N'$ is now more complicated and depends on $a$ (in relation to the expansion parameter for $\beta =1$ see too \cite{Ma12}), and the details of the polynomials in the functional forms of the expansion coefficients differ from those of the Gaussian case. At the hard edge, which is obtained by the scaling $\lambda_l \mapsto x_l/4N$, a recent result from
  \cite{Fo24} gives that particular gap probabilities permit expansion in powers of $N_{\rm L}^{-2}$, where $N_{\rm L} :=
  N + a/\beta$, although no results were obtained in relation to the functional form of the coefficients. To leading order, this result was first established in the earlier work \cite{FT19}, which also contains computation of the functional form of some corrections; in relation to the leading order correction at the hard edge with $\beta =2$ we reference too  \cite{EGP16,Bo16,PS16,HHN16}. Analogous results are also known at the hard of the Jacobi $\beta$ ensemble (see \cite[Prop.~4.6 and Remark 4.3.1]{Fo24}, \cite[Appendix A of arXiv version]{FLT21}, \cite{Wi24} for results relating to general $\beta$, and \cite{MMM19} in relation to $\beta = 2$. After hard edge scaling the correlations and spacing distributions no longer depend on $N$, but do depend on the Laguerre parameter, $a$ say. Asymptotic expansions with respect to $a$ and the re-centering of the position variables in relation to the hard edge to soft edge transition are of interest in the context of the length of the longest increasing subsequence problem \cite{BF03}. Results on the optimal expansion parameter and the resulting functional forms are given in \cite{BJ13,FM23,Bo24x,Bo24y}.

  The expansion in powers of $1/N$ for the bulk scald correlation functions about $\theta_l = 0$ of the circular Jacobi ensemble, specified by the eigenvalue PDF proportional to
  \begin{equation}\label{D4} 
  \prod_{j=1}^N e^{-q \theta_j}
  |1 - e^{i \theta_j} |^{\beta p}
  \prod_{1 \le j < k \le N} | e^{i \theta_k} - e^{i \theta_j} |^\beta, \quad \theta_l \in [0,2 \pi)
  \end{equation} 
  has been considered in the work \cite{FLT21}. There, it was established for $\beta = 1, 2$ and 4 in the case of the general $k$-point correlation function, and for $\beta$ even in the case of the density that in the variable $N_{\rm cJ} :=
  N + p$ the leading correction is proportional to $N_{\rm cJ}^{-2}$.

  \subsection{Summary of results}
  In Section \ref{S2a} the asymptotic expansion
  of the bulk scaled spacing distribution generating function $\mathcal P_N^{\rm CUE}(2 \pi s/N;\xi)$
  as defined by (\ref{4.2c})
   is established to be in powers of $N^{-2}$ to all orders; see (\ref{4.3b}).
  With the leading term denoted $\mathcal P_{0,\beta = 2}^{\rm bulk}(s;\xi)$, and the first correction equal to $N^{-2} \mathcal P_{1,\beta = 2}^{\rm bulk}(s;\xi) $, the functional form $\mathcal P_{1,\beta = 2}^{\rm bulk}(s;\xi)$ is shown,
  using a $\sigma$-Painlev\'e V characterisation,
  to be related to $\mathcal P_{0,\beta = 2}^{\rm bulk}(s;\xi)$ by the simple differential relation (\ref{5.1b}).
  The consequence of this result in the context of the empirical Riemann zeros spacing distribution at large height, and that of their thinning, is noted in
  Remark \ref{R2.1}.2.
  In Section \ref{S2} the explicit form of the $2 \times 2$ matrix correlation kernel specifying the Pfaffian point process form of the $\beta = 1$ and 4 cases of the circular $\beta$ ensemble is used to deduce that the large $N$ expansion of the general $n$-point correlations are in powers of $1/N^2$. This is the analogue of the CUE result (\ref{4.1}), following from an expansion of the (scalar) correlation kernel, and lifts too to imply an analogous asymptotic expansion of the spacing distribution generating functions. Moreover, the first correction in the asymptotic expansion is shown, again using $\sigma$-Painlev\'e characterisations, to be related to the leading term by a simple differential relation ((\ref{5.1e}) in the case $\beta =1$ and (\ref{5.1g}) in the case $\beta =4$).
  Also considered is the large $N$ expansion of the 
  structure function (this is essentially the Fourier transform  of the truncated two-point correlation function). Proposition \ref{P3.1} establishes an asymptotic expansion in powers of $N^{-2}$, and in (\ref{X6}) differential identities are given relating the first and second corrections to the limiting functional form.
Section \ref{S3} makes use of known expressions in terms of generalised hypergeometric functions based on Jack polynomials (see e.g.~\cite[Ch.~13]{Fo10})
 to establish that each $n$-point correlation with $\beta$ even is an even function of $N$. For the two-point function in this setting, an asymptotic expansion in inverse powers of $N$ can be established, which must therefore only involve even powers. In addition, it is shown that the leading 
correction in the case of the two-point correlation function is related to the limiting functional form via a simple differential identity of the form already established for $\beta =1,2$ and 4.  
A recursion identity for multiple integrals, relevant to expressing  higher order corrections in the $1/N^2$ expansion of the two-point correlation for the $\beta$ even circular ensemble to derivative operations of the limiting functional form, is given in Appendix A. Appendix B relates to the functional form of $\xi^k$ in (\ref{4.2}) for the circular $\beta$ ensemble with bulk scaling variables as in (\ref{4.3}), to leading order in $s$ but with $N$ general. This allows for determination of the respective terms in the $\beta$ generalisation of the power series expansion (\ref{1.15}), and thus provides data on the relation of the $1/N^2$ and $1/N^4$
corrections in relation to the limiting values.

\section{Bulk scaled spacing distribution for $\beta =2$}\label{S2a}
Let $p_N^{(\cdot)}(k;x)$ denote the probability that in the ensemble $(\cdot)$, given there is an eigenvalue at the origin, there is an eigenvalue at $x$ with exactly $k$ eigenvalues in between. Suppose that $(\cdot)$ is rotationally invariant and has been bulk scaled so that the eigenvalue density is unity. Analogous to (\ref{4.2}) we have   the generating function expansion\footnote{For $0 < \xi \le 1$, $\xi \mathcal P_N^{(\cdot)}(x;\xi)$ has the further interpretation as the spacing distribution in the situation that each eigenvalue is deleted with probability $(1-\xi)$ uniformly at random; see e.g.~\cite[\S 3.3]{BFM17}. 
We remark too of an alternative occurrence of the gap probability generating function
(\ref{4.1}) in random matrix theory beyond its literal interpretation.
Thus one has that the so-called power spectrum statistic, defined as the Fourier sum associated with the sequence of covariances for the level displacements, can be written in terms of (\ref{4.1}) with $\xi = 1 -z$, and $z$ on the unit circle in the complex plane \cite{RK23}. This is of much significance in the large $N$ analysis of this statistic for the CUE \cite{ROK20,RK23}, and its circular $\beta$ ensemble generalisation \cite{FW24}. }
\begin{equation}\label{4.2c}
 \mathcal P_N^{(\cdot)}(x;\xi)
 := \sum_{k=0}^{N-2} (1 - \xi)^k
  p_N^{(\cdot)}(k;x)
 =  \rho_{(2)}^{(\cdot)}(0,x)+
 \sum_{k=1}^{N-2} {(-\xi)^k \over k!} \int_0^x d x_1 \cdots
 \int_0^x d x_k \,  \rho_{(k+2)}^{(\cdot)}(0,x,x_1,\dots,x_k),
  \end{equation}
which moreover is related to (\ref{4.2}) by
\begin{equation}\label{4.2d}
 \mathcal P_N^{(\cdot)}(x;\xi) = {1 \over \xi^2}
 {d^2 \over dx^2} \mathcal E_N^{(\cdot)}((0,x);\xi),
  \end{equation}
  where again bulk scaling (with eigenvalue density unity) is assumed.

  It follows from the second expression in (\ref{4.2c}) and the $1/N^2$ expansion of the correlation functions that the bulk scaled generating function $(2 \pi /N)^2 \mathcal P_N^{\rm CUE}(2 \pi s/N;\xi)$ for large $N$, as for 
 $\mathcal E_N^{\rm CUE}(2 \pi s/N;\xi)$ in 
 (\ref{4.3}), has the asymptotic expansion
 \begin{equation}\label{4.3b}
 \Big ( {2 \pi  \over N} \Big )^2 \mathcal P_N^{\rm CUE}(2 \pi s / N;\xi) \sim
\mathcal P^{\rm bulk}_{0,\beta = 2}(s;\xi) + \sum_{l=1}^\infty {1 \over N^{2l} } \mathcal P^{\rm bulk}_{l,\beta = 2}(s;\xi).
  \end{equation}
Applying (\ref{4.2d}) to (\ref{1.15}) gives the small $s$ power series
\begin{align}
 &{ \mathcal P_{0,\beta=2}^{\rm bulk}(s;\xi) } =\frac{ \pi ^2 }{3}  s^2 -\frac{2  \pi ^4}{45}   s^4
   +\frac{ \pi ^6}{315}  s^6 
   -\frac{\pi ^6
   }{4050}  \xi s^7 
   -\frac{2 \pi ^8 
   }{14175}  s^8 +\frac{11 \pi ^8 }{496125}  \xi s^9+ {\rm O}(s^{10})
 \label{2.4}\\
&{ \mathcal P_{1,\beta=2}^{\rm bulk}(s;\xi) } =
-\frac{ \pi ^2}{3}  s^2 +\frac{\pi ^4}{9}   s^4
   -\frac{2\pi ^6}{135}  s^6 
   +\frac{\pi ^6}{675} 
  \xi s^7 
   +\frac{\pi ^8 }{945} s^8 
   -\frac{121 \pi ^8 }{595350}  \xi s^9   
   + {\rm O}(s^{10})
   \label{2.5} \\
   &{ \mathcal P_{2,\beta=2}^{\rm bulk}(s;\xi) } =
-\frac{ \pi ^4}{15}  s^4 +\frac{\pi ^6}{45}  s^6  -\frac{ \pi ^6}{450} \xi s^7 
   -\frac{2 \pi ^8 }{675}  s^8 
   +\frac{44 \pi ^8} {70875} 
   \xi s^9    + {\rm O}(s^{10}). \label{2.6}
   \end{align}
In these expansions, according to the second expression in (\ref{4.2c}), the terms independent of $\xi$ are obtained from the small $X-Y=:s$ expansion of the RHS of (\ref{3.3}). In relation to the latter, with the notation
$$
\Big ( {2 \pi \over N} \Big )^2 \rho_{(2)}^{\rm CUE}
 ( 2 \pi s/N,0) \sim
\rho_{(2),0,\beta=2}^{\rm bulk}(s,0) + {1 \over N^2} \rho_{(2),1,\beta=2}^{\rm bulk}(s,0) +
{1 \over N^4} \rho_{(2),2,\beta=2}^{\rm bulk}(s,0) + \cdots
$$
as is consistent with (\ref{4.1}),
we observe the simple differential functional relation 
 \begin{equation}\label{5.1}
\rho_{(2),1,\beta=2}^{\rm bulk}(s,0) 
= - {1 \over 12} {d^2 \over d s^2} \bigg (
s^2 \rho_{(2),0,\beta=2}^{\rm bulk}(s,0) \bigg ).
\end{equation}

Remarkably this same relation holds true between $\mathcal P^{\rm bulk}_{0,\beta = 2}(s;\xi)$ and  $\mathcal P^{\rm bulk}_{1,\beta = 2}((0,s);\xi)$ (preliminary evidence is the validity of such a relation between the power series (\ref{2.4}) and (\ref{2.5})), which then is the bulk scaling analogue of (\ref{D2}).

\begin{prop}\label{P2.1}
For the asymptotic expansion (\ref{4.3b}) of the
generating function (\ref{4.2c}) in the case of the CUE we have
\begin{equation}\label{5.1b}
 \mathcal P_{1,\beta=2}^{\rm bulk}(s;\xi) =
 -{1 \over 12} {d^2 \over ds^2}
 \bigg ( s^2 \mathcal P_{0,\beta=2}^{\rm bulk}(s;\xi) \bigg ).
\end{equation}
Equivalently, with reference to the asymptotic expansion (\ref{4.3}),
\begin{equation}\label{5.1c}
 {\mathcal E}_{1,\beta=2}^{\rm bulk}(s;\xi) =
 -{1 \over 12}
  s^2  {d^2 \over ds^2} \mathcal E_{0,\beta=2}^{\rm bulk}(s;\xi). 
\end{equation}
\end{prop}

\begin{proof}
Making use of (\ref{4.2d}) in  (\ref{4.3}) and comparing with (\ref{4.3b}) shows the equivalence between (\ref{5.1b}) and (\ref{5.1c}). We will henceforth consider (\ref{5.1c}). 

Let $\mathbb K_s$ denote the integral operator on $(0,s)$ with kernel $K_\infty(x,y)$ as specified in (\ref{2.1a}). It is a classical result in random matrix theory \cite[\S 6.3]{Me04}, \cite[\S 9.3]{Fo10} that 
\begin{equation}\label{5.1d}
\mathcal E_{0,\beta=2}^{\rm bulk}(s;\xi) = \det
( \mathbb I - \xi \mathbb K_s),
\end{equation}
where here $\det(\cdot)$ refers to the Fredholm determinant, and $\mathbb I$ is the identity operator. Let $\mathbb L_s$ denote the integral operator on $(0,s)$ with kernel 
\begin{equation}\label{KL}
L_\infty(x,y) : = (\pi (x-y)/6) \sin(\pi (x - y)). 
\end{equation}
Less well known, but nonetheless available in the literature \cite{FM15,BFM17}, is the operator expression
\begin{equation}\label{5.1d1}
\mathcal E_{1,\beta=2}^{\rm bulk}(s;\xi) = - \det
( \mathbb I - \xi \mathbb K_s) {\rm Tr} (
(\mathbb I - \xi \mathbb K_s)^{-1} \xi \mathbb L_s).
\end{equation}
Thus one strategy to prove (\ref{5.1c}) would be to establish the identity implied by the forms in (\ref{5.1d}) and (\ref{5.1d1}). In the case that  $\mathbb K_s$ in (\ref{5.1d}) is the Airy kernel and $\mathbb L_s$ originates from the $\beta = 2$ soft edge analogue of the expansion (\ref{3.3}) of the kernels for the Gaussian and Laguerre unitary ensembles, such identities have been systematically considered in \cite{Bo25a}.

An alternative strategy, used to establish identities in the same general class as (\ref{5.1c}) (i.e.~linking a leading order distribution function, itself permitting a Fredholm determinant form, to its leading order correction by a differential relation) makes use of alternative expressions in terms of sigma Painlev\'e transcendents \cite{FPTW19}, \cite{FM23}. In this regards, one recalls that it is pioneering result in the theory of integral systems due to the Kyoto group \cite{JMMS80} that an alternative to (\ref{5.1d}) is the so-called $\tau$-function formula for a particular Painlev\'e V system
\begin{equation}\label{KS}
\mathcal E_{0,\beta=2}^{\rm bulk}(s;\xi) = \exp \int_0^{\pi s} {\sigma_0 (t;\xi) \over t} \, dt,
\end{equation}
where $\sigma_0$ satisfies the 
particular $\sigma$-Painlev\'e V
differential equation
\begin{equation}\label{d1y}
(t \sigma_0'')^2 + 4 (t \sigma_0' - \sigma_0)(t \sigma_0' - \sigma_0 + (\sigma_0')^2) = 0
\end{equation}
with small $t$ boundary condition 
\begin{equation}\label{d2y}
\sigma_0(t;\xi) = - {\xi \over \pi} t - {\xi^2 \over \pi^2} t^2 + {\rm O}(t^3).
\end{equation}
Some years later, Forrester and Mays \cite{FM15} obtained that
\begin{equation}\label{KSa}
\mathcal E_{1,\beta=2}^{\rm bulk}(s;\xi) =  
\mathcal E_{0,\beta=2}^{\rm bulk}(s;\xi) 
 \int_0^{\pi s} {\sigma_1(t) \over t} \, dt,
\end{equation}
where $\sigma_1(t)$ satisfies a second order linear differential equation
\begin{equation}\label{Sig2}
A(t) \sigma_1'' + B(t) \sigma_1' + C(t) \sigma_1 = D(t),
\end{equation}
where the functions $A(t), \dots, D(t)$ depend on $\sigma_0$ and possibly its first and second derivative (e.g.~$A(t)  = 2 t^2 \sigma_0''$ --- for the details of the others see \cite[Prop.~3.1]{FM15}), subject to the small $t$ boundary condition
\begin{align}\label{bcs}
\sigma_1(t;\xi)  = - \Big ( t^4 {\xi^2 \over 9 \pi^2} + t^5 {5 \xi^3 \over 36 \pi^3} + {\rm O}(t^6) \Big ).
\end{align}

Substituting (\ref{KS}) and (\ref{KSa}) in (\ref{5.1c}) shows that the latter is valid provided
\begin{equation}\label{Sig3}
\sigma_1 = - {1 \over 12} \Big (  2 t \sigma_0 \sigma_0' + t^2 \sigma_0'' \Big ).
\end{equation}
Substituting a suitable extension of the boundary condition (\ref{d2y}) \cite[Eq.~(3.20)]{FM15} in this equation gives consistency with (\ref{bcs}). To verify (\ref{Sig3}) itself, as done in similar situations in \cite{FPTW19} and \cite{FM23}, the basic idea is to show that the RHS satisfies (\ref{Sig2}), using knowledge of the fact that $\sigma_0$ satisfies (\ref{d1y}). For this purpose, one begins by noting that differentiating the latter shows
$$
t^2 \sigma_0''' + t \sigma_0'' + 6 t (\sigma_0')^2 + 4 t^2 \sigma_0' -4 \sigma_0 (t + \sigma_0') = 0,
$$
which allows for the dependence on $\sigma_0'''$ in the resulting equation to be eliminated. A further differentiation shows that the fourth derivative as results from
differentiating (\ref{Sig3}) twice
can similarly be eliminated, leaving an identity in $\{ \sigma_0, \sigma_0',\sigma_0'',t \}$. With the help of computer algebra, the latter can then be checked upon direct use of (\ref{d1y}).
\end{proof}

\begin{remark}  \label{R2.1} ${}$ \\
1.~A celebrated approximation to 
$\mathcal P_{0,\beta =2}(s;\xi)|)|_{\xi = 1}$ is
the Wigner surmise $p_{\beta = 2}^{\rm W}(s) =
{32 s^2 \over \pi^2} e^{-4 s^2/\pi}$; see e.g.~\cite{Ha00}. Substituting in (\ref{5.1b}) with $\xi = 1$ and graphically comparing against the exact functional form 
of $\mathcal P_{1,\beta =2}(s;\xi)|)|_{\xi = 1}$
obtained by substituting instead
$\mathcal P_{0,\beta =2}(s;\xi)|)|_{\xi = 1} = {d^2 \over d s^2} \det (\mathbb I - \mathbb K_{s})$ and where $\mathbb K_{s}$ is as in
(\ref{5.1d}), with the Fredholm determinant herein computed according to Bornemann's method \cite{Bo08,Bo10}, shows the same high level of graphical accuracy for the approximate leading correction term as does the Wigner surmise for the limiting distribution; see Figure \ref{Fig1}. \\
2,~For a fixed $N$, the quantity $(2 \pi/N)^2 \mathcal P_N^{\rm CUE}(2 \pi s / N;\xi)$ can be computed empirically from numerically generated eigenvalues of
CUE matrices for all $0 < \xi \le 1$. Thus for $\xi = 1$ this is just the spacing distribution for consecutive eigenvalues, while for $0< \xi < 1$ it is the spacing
distribution after deleting each eigenvalue with probability $(1-\xi)$ \cite{BP04}. Subtracting this from  $\mathcal P_{0,\beta=2}^{\rm bulk}(s;\xi)$ and scaling by
$N^2$ will then give the graph of $\mathcal P_{1,\beta =2}(s;\xi)|)$ up to corrections of order $1/N^2$.
In accordance with an hypothesis first put forward by Keating and Snaith \cite{KS00a} and further developed in \cite{BBLM06,FM15,BFM17}, 
this same effect holds true for the empirical determination of  $\mathcal P_N^{(\cdot)}(x;\xi)$ for the Riemann zeros at large height. Indeed one can observe (up to scaling and statistical fluctuations)
the graphical form of Figure \ref{Fig1} in the analysis of the latter carried out in \cite{BBLM06,FM15} --- see in particular \cite[Fig.~10]{FM15}.
3.~The structure of (\ref{5.1b}) shows immediately that $\int_0^\infty s^j \mathcal P_{1,\beta=2}^{\rm bulk}(s;\xi) \, ds = 0$ for $j=0,1$. To anticipate this, note
from the finite $N$ definition of $\mathcal P_N^{(\cdot)}(x;\xi)$
(\ref{4.2c}), that the requirements that for the normalisation and the first moment $\int_0^\infty x^j p_N^{(\cdot)}(k;x) \, dx = 1 \, (j=0)$ and $k \, (j=1)$ (the latter with bulk scaling, density unity, imposed) tell us $\int_0^\infty s^j  \mathcal P_{1,\beta=2}^{\rm bulk}(s;\xi) \, ds$ for $j=0,1$ is independent of 
$N$. \\
4.~In the extension of the expansion (\ref{D2}) as an asymptotic series in powers of $N^{-2/3}$,
the explicit form of the functional form at order $N^{-4/3}$ is given as
$\sum_{l=1}^4 p_l(t) {d^l \over d t^l} E_{\beta = 2}^{\rm soft}((t,\infty);\xi)$ for certain (low order) polynomials $\{ p_l(t) \}$, with analogous functional forms conjectured at higher order (made explicit at order $N^{-2}$) \cite{Bo24,Bo25b}.
However, in relation to ${\mathcal P}^{\rm bulk}_{2,\beta=2}(s;\xi)$ in (\ref{4.3b}), we have yet to find evidence for a generalisation of (\ref{5.1b}) using as data the power series in
(\ref{2.4}) and (\ref{2.6}). 
For the two-point correlation function which is the $\xi = 0$ case of $\mathcal P_{0,\beta =2}(s;\xi)$, from the explicit functional form exhibited by extending the expansion (\ref{3.3}) to O$(N^{-4})$, one can check that
        \begin{equation}\label{R1z}
\rho_{(2),2,\beta=2}^{\rm bulk}(s,0) 
= - {(\pi s)^2  \over 60} {d^2 \over d s^2} \bigg (
s^2 \rho_{(2),0,\beta=2}^{\rm bulk}(s,0) \bigg );
        \end{equation}
        cf.~(\ref{5.1}). However, as a candidate for linking the O$(N^{-4})$ term for the spacing generating function to its limiting form,
        this is incompatible with the 
        $\xi$-dependent terms in (\ref{2.4}) and (\ref{2.6}).
In Appendix B we 
provide further data on the relation between the powers series (\ref{2.4}), and (\ref{2.5}) and (\ref{2.6}), for the leading small-$s$ term at each power of $\xi$.

\end{remark}

\begin{figure}
\centering
\includegraphics[width=3.5in]{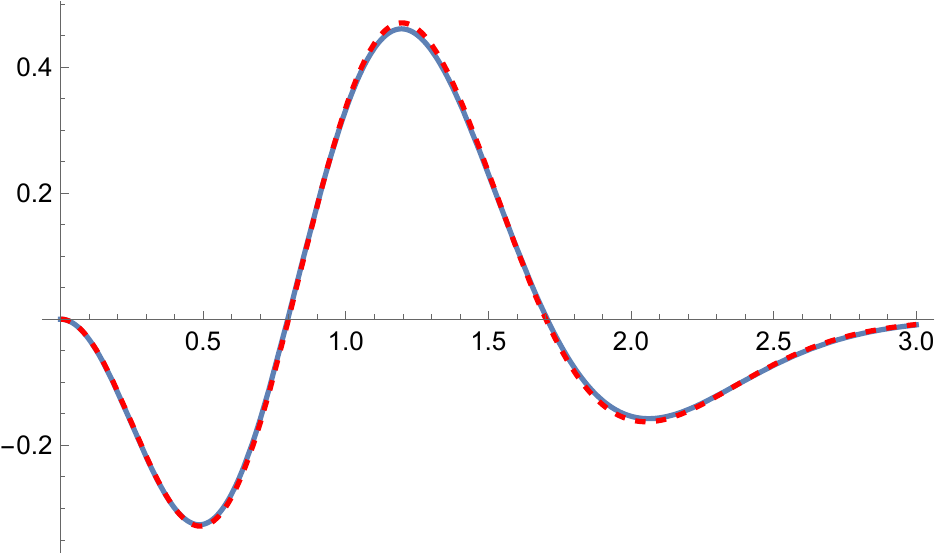}
\caption{[color-on-line] Exact functional form of $\mathcal P_{1,\beta =2}(s;\xi)|)|_{\xi = 1}$ (solid blue line) plotted together with the approximation obtained by substituting the Wigner surmise for $\mathcal P_{0,\beta =2}(s;\xi)|)|_{\xi = 1}$ in (\ref{5.1b}) (dashed red line).}
\label{Fig1}
\end{figure}
    
\section{Bulk scaling expansion of the correlations for $\beta = 1$ and $\beta = 4$}\label{S2}
\subsection{Pfaffian point process viewpoint}
The eigenvalue PDF \eqref{1.1A} with $\beta=1$ and $4$ corresponds to  circular orthogonal and symplectic ensembles respectively, denoted COE and CSE, which are known to form Pfaffian point processes (see e.g.~\cite[Ch.~6]{Fo10}). Specifically, the $k$-point correlation functions have the Pfaffian form
\begin{align*}
    \rho_{(n)}^{(\cdot)}(\theta_1,\dots,\theta_n)=\pf [K_N^{(\cdot)}(\theta_j,\theta_k)]_{j,k=1,\dots,n},
\end{align*}
where the kernels are given by $2\times 2$ matrices
\begin{align}\label{K14}
    K^{{\rm COE}}_N(\theta,\phi)=\left[\begin{array}{cc}
        J_N(\theta-\phi) & S_N(\theta-\phi) \\
        -S_N(\theta-\phi) & -D_N(\theta-\phi)
    \end{array}\right],\quad K^{{\rm CSE}}_N(\theta,\phi)=\frac{1}{2}\left[\begin{array}{cc}
        I_{2N}(\theta-\phi) & S_{2N}(\theta-\phi) \\
        -S_{2N}(\theta-\phi) & -D_{2N}(\theta-\phi)
    \end{array}\right].
\end{align}
Here the entries are given by functions
\begin{align*}
    &S_N(\theta)=K_N^{\rm CUE} (\theta, 0) = {1 \over 2 \pi} {\sin (N\theta/2) \over \sin(\theta/2)},\quad D_N(\theta)=\frac{d}{d\theta}S_N(\theta),\\
    & I_N(\theta)=\int_0^\theta S_N(\phi)d\phi,\quad J_N(\theta)=I_N(\theta)-\epsilon_N(\theta),
\end{align*}
where $\epsilon_N(\theta)$ depends on the parity of $N$ \cite[\S 11.3]{Me04}
\begin{align*}
    &\epsilon_{2n}(\theta)=\left\{\begin{array}{cc}
        \frac{1}{2}(-1)^m,& 2\pi m<\theta<2\pi(m+1),   \\
        0 & \theta=2\pi m . 
    \end{array}\right. \\
    &\epsilon_{2n+1}(\theta)=\left\{\begin{array}{cc}
        m+\frac{1}{2},& 2\pi m<\theta<2\pi(m+1),   \\
        m & \theta=2\pi m . 
    \end{array}\right.
\end{align*}
In the bulk scaling, unlike $\frac{1}{N}S_N({2 \pi X\over N})$, the scaled functions $D_N,I_N$ and $J_N$ are odd functions in $N$. Nevertheless, the $n$-point correlation functions are still even in $N$ which can be seen through the fact that the functions $I_N$($J_N$)/$D_N$ appears exclusively only in the odd/even rows of the Pfaffian. 
Of particular interest are the functional forms
of the leading correction in comparison to the limit itself (these functional forms being classical \cite{Me04}) in the case of the two-point function.

\begin{prop}
For $\beta =1,4$, and for certain functional forms $\{ \rho_{(n),l,\beta }^{\rm bulk} \}_{l=0,1,\dots,}$, we have
\begin{equation}\label{4.1B}
\Big ( {2 \pi \over N} \Big )^n
\rho_{(n),\beta}\bigg ( {2 \pi X_1 \over N},\dots,
 {2 \pi X_n \over N} \bigg ) \sim
 \rho_{(n),0,\beta }^{\rm bulk}(X_1,\dots,X_n) +
 \sum_{l=1}^\infty {1 \over N^{2l}} \rho_{(n),l,\beta }^{\rm bulk}(X_1,\dots,X_n).
 \end{equation}
 Specifically, in the case $n=2$ of the two-point correlation
 \begin{align*}
&  \rho_{(2),0,\beta=1 }^{\rm bulk}(X,0)=
1-\bigg(\frac{\sin \pi X}{\pi X}\bigg)^{2}+\frac{(\sin\pi X-\pi X\cos\pi X)(\pi {\rm sgn}(X)-2{\rm Si}(\pi X))}{2\pi^2X^2} \\
& \rho_{(2),1,\beta=1 }^{\rm bulk}(X,0)=
\frac{1}{{12 }}\left(-4 \sin ^2(\pi  X)-\frac{2 (\sin (\pi  X)-\pi  X \cos (\pi 
   X))^2}{\pi ^2 X^2}\right.\\
   &\hspace*{3cm} -(\pi{\rm sgn}(X) -2 {\rm Si}(\pi  X)) (\sin (\pi  X)+\pi  X \cos
   (\pi  X))\bigg) \\
& \rho_{(2),0,\beta=4 }^{\rm bulk}(X,0)=
\frac{1}{4}\bigg(1-\bigg(\frac{\sin \pi X}{\pi X}\bigg)^{2}\bigg)-\frac{{\rm Si} (\pi X) (\sin \pi X-\pi X\cos \pi X)}{4(\pi X)^2} \\
&\rho_{(2),1,\beta=4 }^{\rm bulk}(X,0)=
\frac{1}{96}\bigg(1-3\sin^2\pi X+\frac{\sin^2\pi X}{(\pi X)^2}-\frac{\sin 2\pi X}{\pi X}\\
 &\hspace*{3cm} -{\rm Si (\pi X)}(\pi X\cos \pi X+\sin\pi X)\bigg),
   \end{align*}
   where $
    {\rm Si}(x):=\int_0^x \frac{\sin t}{t}dt$. 
\end{prop}

A differential relation relating the functional form of the leading correction to the limiting functional form analogous to (\ref{5.1}) can be observed.
\begin{cor}
For $\beta = 1,4$ we have
 \begin{equation}\label{5.1B}
\rho_{(2),1,\beta}^{\rm bulk}(s,0) 
= c_\beta {d^2 \over d s^2} \bigg (
s^2 \rho_{(2),0,\beta}^{\rm bulk}(s,0) \bigg )
\end{equation}
with $c_1 = - {1 \over 6}, c_4 = - {1 \over 24}$.
\end{cor}

\begin{proof} This is verified using computer algebra from the explicit functional forms. \end{proof}

\subsection{Spacing distributions}
We have from \cite[Cor.~5.12]{BFM17} the $\beta = 1$ spacing distribution expansions analogous to (\ref{2.4}) and (\ref{2.5}) 
\begin{gather}
 \begin{multlined}\label{3.4Z}
\mathcal P_{0,\beta=1}(s;\xi) = \frac{1}{6} \pi ^2  s   -\frac{1}{60} \pi ^4  s^3  -\frac{1}{270} \pi ^4(\xi
   -2)   s^4 +\frac{\pi ^6  s^5  }{1680}  +\frac{\pi ^6
   (\xi -2)   s^6 }{4725}  \\*[2mm] - \frac{\pi ^8  s^7 }{90720} +\frac{\pi ^8  (\xi
   -2) (3 \xi -32) s^8  }{5292000} +\frac{\pi ^{10}  s^9 }{7983360}   + {\rm O}(s^{10}),
  \end{multlined}\\
  \begin{multlined}\label{3.5Z}
  \mathcal P_{1,\beta=1}(s;\xi) = -\frac{1}{6} \pi ^2  s +\frac{1}{18} \pi ^4  s^3 +\frac{1}{54} 
   \pi ^4 (\xi -2)  s^4-\frac{1}{240} \pi ^6   s^5  -\frac{4 \pi ^6
    (\xi -2)  s^6}{2025}\\*[2mm] +\frac{\pi ^8  s^7}{7560}-\frac{\pi ^8
    (\xi -2) \left(3 \xi -32 \right) s^8}{352800}-\frac{\pi^{10}  s^9}{435456} + {\rm O}\left(s^{10}\right).
 \end{multlined}  
 \end{gather}
To the order shown, these spacing distribution generating functions, 
and the gap probability generating functions $ \mathcal E_{0,\beta=1}^{\rm bulk}(s;\xi), 
\mathcal E_{1,\beta=1}^{\rm bulk}(s;\xi)$ too,
are related by the analogues of (\ref{5.1b})
and (\ref{5.1c})
\begin{equation}\label{5.1e}
 \mathcal P_{1,\beta=1}^{\rm bulk}(s;\xi) =
 -{1 \over 6} {d^2 \over ds^2}
 \bigg ( s^2 \mathcal P_{0,\beta=1}^{\rm bulk}(s;\xi) \bigg ), \quad  \mathcal E_{1,\beta=1}^{\rm bulk}(s;\xi) =
 -{s^2 \over 6} {d^2 \over ds^2}
  \mathcal E_{0,\beta=1}^{\rm bulk}(s;\xi),
\end{equation}
the first of which with $\xi = 0$ reduces to
the case $\beta =1$ of (\ref{5.1B}).
In fact these identities (since they are equivalent to each other, it suffices to consider only the second) can be established to be true in general.

For this purpose, we will use that fact that
analogous to (\ref{5.1d}) and (\ref{5.1d1}), the gap probability generating functions in (\ref{5.1e}) admit forms relating to certain integral operators.
As already commented, the circular ensemble for $\beta = 1$ is the circular orthogonal ensemble, which gives rise to a Pfaffian point process. These imply matrix Fredholm determinant forms (specifically the kernel of the corresponding integral operator is a $2 \times 2$ matrix) for the generating functions. We require instead a relation to certain scalar integral operators. This comes about due to the generating function identity
(\cite[Eq.~(8.150) $N$ even case]{Fo10}, \cite[Prop.~5.1 \& Eq.~(5.10)]{BFM17})
  \begin{equation}\label{2.31}
\mathcal  E^{{\rm COE}_N}((0,\phi);\xi) = {(1 - \xi)  \mathcal  E^{O^\nu(N+1)}((0,\phi/2);\hat{\xi})   + \mathcal  E^{O^{-\nu}(N+1)}((0,\phi/2);\hat{\xi})  \over 2 - \xi},
  \end{equation} 
  where $\nu = (-)^N$ and $\hat{\xi} := 2 \xi - \xi^2$. In (\ref{2.31}) the notation $O^+(n)$ and $O^-(n)$ refers to the eigenvalues in $(0,\pi)$ of these classical matrix groups, while COE${}_N$ refers to the $N \times N$ circular orthogonal ensemble.  The eigenvalues of the classical groups orthogonal ensembles form determinantal point processes  with respective kernels (see e.g.~\cite[\S 5.5.2]{Fo10})
  \begin{equation}\label{Kpm}
  K^{N,\mp}(x,y) := K_N^{\rm CUE}(x,y) \mp K_N^{\rm CUE}(x,-y).
 \end{equation} 

 Expanding the bulk scaling of the RHS of (\ref{2.31}) for large $N$ according to (\ref{3.2}) gives
 \begin{gather}
 \begin{multlined}\label{3.6}
\mathcal E_{0,\beta=1}((0,s);\xi) = {1 \over 2 - \xi} \bigg (  (1 - \xi)
     \det(\mathbb I - \bar{\xi} \mathbb K_{s/2}^-)   
      + \det(\mathbb I - \bar{\xi} \mathbb K_{s/2}^+)  
      \bigg ),
     \end{multlined} \\
 \begin{multlined}\label{3.7}
      \mathcal E_{1,\beta=1}((0,s);\xi) = {1 \over 2 - \xi} \bigg (  (1 - \xi)
     \det(\mathbb I - \bar{\xi} \mathbb K_{s/2}^-) {\rm Tr}(( \mathbb I - \xi \mathbb K_{s/2}^-)^{-1}
     \xi \mathbb L_{s/2}^-)  \\
      + \det(\mathbb I - \bar{\xi} \mathbb K_{s/2}^+)  
     {\rm Tr}(( \mathbb I - \xi \mathbb K_{s/2}^+)^{-1}
     \xi \mathbb L_{s/2}^+) \bigg ).
     \end{multlined}
     \end{gather}
 Here $\mathcal K_s^\pm$ and $\mathcal L_s^\pm$ are 
 the integral operators on $(0,s)$   with kernels 
 $$
 K_\infty^{\pm} (X,Y) := K_\infty(X,Y) \pm K_\infty(X,-Y), \qquad
 L_\infty^{\pm} (X,Y) := L_\infty(X,Y) \pm L_\infty(X,-Y),
 $$
 where $K_\infty, L_\infty$ are given by (\ref{2.1a}) and (\ref{KL}) respectively. Let us denote $\det(\mathbb I - {\xi} \mathbb K_{s/2}^\pm)$ by $\mathcal E_{0}^\pm(s;\xi)$ and 
$\det(\mathbb I - {\xi} \mathbb K_{s/2}^\pm) {\rm Tr}(( \mathbb I - \xi \mathbb K_{s/2}^\pm)^{-1}
     \xi \mathbb L_{s/2}^\pm)$ by
  $\mathcal E_{1}^\pm(s;\xi)$. 
  We can establish the following differential identity between these quantities, which is sufficient  for the validity of  (\ref{5.1e}).

  \begin{prop}\label{P3.1x}
      We have
      \begin{equation}\label{5.1f}
   \mathcal E_{1}^{\pm}(s;\xi) =
 -{s^2 \over 6} {d^2 \over ds^2}
  \mathcal E_{0}^{\pm}(s;\xi).
\end{equation}
      \end{prop}

      \begin{proof}
As with Proposition \ref{P2.1}, our strategy is to make use of Painlev\'e transcendent characterisations of the quantities in (\ref{5.1f}). Specifically, in \cite[Prop.~5.10]{BFM17} $ \mathcal E_{0,\beta=1}^{\pm}(s;\xi)$ was expressed as  particular $\sigma$PIII$'$ $\tau$-functions, invovling certain transcendents
$f_0^\mp$. The structure of $ \mathcal E_{1,\beta=1}^{\pm}(s;\xi)$ was shown to be analogous to (\ref{KSa}), and thus involving transcendents $f_1^\mp$ specified by the solution of a second order linear differential equation analogous to (\ref{Sig2}) subject to certain boundary conditions. The identity (\ref{5.1f}) is then seen to be equivalent to an identity analogous to (\ref{Sig3}) expressing $f_1^\mp$ in terms of $f_0^\mp$. This in term can be established following the procedure used to establish (\ref{Sig3}).
      \end{proof}

    The identity (\ref{5.1f}) also has relevance to the $\beta = 4$ analogue of  (\ref{5.1b}),
 (\ref{5.1c}) and (\ref{5.1e}).

 \begin{prop}
      We have
      \begin{equation}\label{5.1g}
 \mathcal P_{1,\beta=4}^{\rm bulk}(s;\xi) =
 -{1 \over 24} {d^2 \over ds^2}
 \bigg ( s^2 \mathcal P_{0,\beta=4}^{\rm bulk}(s;\xi) \bigg ), \quad  \mathcal E_{1,\beta=4}^{\rm bulk}(s;\xi) =
 -{s^2 \over 24} {d^2 \over ds^2}
  \mathcal E_{0,\beta=4}^{\rm bulk}(s;\xi).
\end{equation}
      \end{prop}

      \begin{proof}
          First, as already commented in relation to (\ref{5.1b}) and
 (\ref{5.1c}), the two identities are equivalent, so it suffices to consider only the second of these. In the case of the CSE (i.e.~the circular $\beta$ ensemble with $\beta = 4$ we have that \cite[Eq.~(5.15)]{Fo10}
 (for context see the recent review \cite[\S 3.3]{Fo24x})
   \begin{equation}\label{2.31z}
\mathcal  E^{{\rm CSE}_N}((0,\phi);\xi) = {1 \over 2} \Big ( \mathcal  E^{O^+(2N+1)}((0,\phi/2);{\xi})   + \mathcal  E^{O^{-}(2N+1)}((0,\phi/2);{\xi}) \Big ).
  \end{equation} 
  Consequently
  \begin{equation}
      \mathcal E_{0,\beta = 4}((0,s);\xi) = {1 \over 2} \Big ( \mathcal  E^{-}_0(s;{\xi})   + \mathcal  E^{+}_0(s;{\xi}) \Big ), \quad
 \mathcal E_{1,\beta = 4}((0,s);\xi) = {1 \over 8} \Big ( \mathcal  E^{-}_1(s;{\xi})   + \mathcal  E^{+}_1(s;{\xi}) \Big ),    
  \end{equation}
 where use has been made of the notation introduced above Proposition \ref{P3.1x}; cf.~the rewrite of (\ref{3.6}) and (\ref{3.7}) using this notation --- note in particular a further factor of $1/4$ appearing on the RHS of $\mathcal E_{1,\beta = 4}((0,s);\xi)$ relative to $\mathcal E_{1,\beta = 1}((0,s);\xi)$, which is traced back to the occurrence of the dimensions of the orthogonal groups on the RHS of (\ref{2.31z}) being $2N+1$ rather than $N+1$ as in (\ref{2.31}).
      \end{proof}

\subsection{The structure function expanded for large $N$}
The structure function (also known as the spectral form factor) $S_{N,\beta}(k)$ is defined in terms of the Fourier coefficients of the truncated two-point correlation $\rho_{(2),N}^T(0,\theta;\beta) = \rho_{(2),N}(0,\theta;\beta) -  ( N / 2 \pi )^2$ according to
 \begin{equation}\label{S1q} 
 S_{N,\beta}(k) = \int_0^{2 \pi} e^{i k \theta} \rho_{(2),N}^T(0,\theta;\beta) \, d \theta + {N \over 2\pi}, \quad k \ge N,
 \end{equation}
 where the additive constant terms ensures that $S_{N,\beta}(0)=0$.
 In the case $\beta = 2$ if follows from (\ref{1.3}) and (\ref{1.4}) that
 \begin{equation}\label{S2q} 
 \rho_{(2),N}^T(0,\theta;\beta) \Big |_{\beta = 2}  = - \Big ( {1 \over 2 \pi} \Big )^2 
 \sum_{p=-N}^{N} (N - |p|) e^{i p \theta}.
 \end{equation}
 Consequently we obtain the now classical (see e.g.~\cite{Ha00}) exact evaluation
 \begin{equation}\label{S3q} 
 S_{N,\beta}(k) \Big |_{\beta = 2} =
 \begin{cases} |k|/2 \pi, & |k| < N \\
 N/2\pi, & |k| \ge N.
 \end{cases}
\end{equation}

Consider next the bulk scaling of the structure function
 \begin{equation}\label{S4q}
 \tilde{S}_{N,\beta}(\tau) := {2 \pi \over N} S_{N,\beta}( \tau N) =
   \int_{-N/2}^{N/2} e^{2 \pi i x \tau } \bigg (
  \Big ({2 \pi \over N} \Big )^2 \rho_{(2),N}^T(0,2 \pi x/N;\beta) \bigg ) \, d x + 1.
\end{equation}  
Recalling (\ref{2.1}) we see that (\ref{S4q}) is, up to the additive constant, the Fourier transform in the Fourier variable $2 \pi \tau$ of the truncated two-point correlation in bulk scaling variables. To obtain the large $N$ form for $\beta =2$, one may trial making use of (\ref{3.3}). However doing so leads to an ill-defined integral at order $1/N^2$. Rather, from the exact result (\ref{S3q}) we see that
\begin{equation}\label{S5q} 
 \tilde{S}_{N,\beta}(\tau) \Big |_{\beta = 2} =
 \begin{cases}  |\tau|, & |\tau| < 1 \\
 1, & |\tau| \ge 1,
 \end{cases}
\end{equation}
and thus as the RHS is independent of $N$ all terms in the corresponding large $N$ expansion in fact vanish except the leading term.

The explicit leading corrections for $\beta=1$ and 4 of Proposition \ref{P3.1}
show that it is similarly not possible to obtain meaningful large $N$ expansions of the bulk scaled structure function by substituting in (\ref{S4q}) the large $N$ expansion of the bulk scaled two-point correlation. Instead one should use knowledge of the exact form of the structure function for these $\beta$ values at finite $N$ \cite{H+96,WF15},
 \begin{equation}\label{S6q} 
2 \pi  S_{N,\beta}(k) \Big |_{\beta = 1} =
 \begin{cases} 2|k| - |k| \Big ( \psi(|k| + (N+1)/2) - \psi((N+1)/2) \Big ), & |k| < N \\ 2 N -
 |k| \Big ( \psi(|k| + (N+1)/2) - \psi(|k| + (-N+1)/2) \Big ), & |k| \ge N,
 \end{cases}
\end{equation}
and
 \begin{equation}\label{S7q} 
2 \pi  S_{N,\beta}(k) \Big |_{\beta = 4} =
 \begin{cases} {|k|\over 2} \Big (1+ {1 \over 2} ( \psi(N+1/2) - \psi(-N+|k|+1/2)) \Big ), & |k| < 2N-1 \\
 N , & |k| \ge 2 N - 1,
 \end{cases}
\end{equation}
where $\psi(z)$ denotes the digamma function.

\begin{prop}\label{P3.1}
    We have that for large $N$
     \begin{equation}\label{S8q} 
  \tilde{S}_{N,\beta}(\tau) \Big |_{\beta = 1} =
 \begin{cases} \displaystyle 2 |\tau| - |\tau| \log (1 + 2 | \tau|) - {|\tau| \over 6 N^2} 
 \Big (  1 - {1 \over  (1 + 2 |\tau|)^2}\Big ) + \cdots 
 & |\tau| \le  1 \\[.3cm]
 \displaystyle 2 -  |\tau|  \log {2 | \tau| + 1 \over 2 | \tau | -1}
 + {|\tau| \over N^2} {4 |\tau| \over 3 (1 - (2 |\tau|)^2)^2} 
+ \cdots & |\tau| \ge  1,
 \end{cases}
\end{equation}
 and
  \begin{equation}\label{S9q} 
  \tilde{S}_{N,\beta}(\tau) \Big |_{\beta = 4} =
 \begin{cases}  \displaystyle  {|\tau| \over 2} - {|\tau| \over 4} \log |1 - |\tau| | +
 {|\tau| \over 96 N^2} \Big ( 1 - {1 \over (|\tau| - 1)^2} \Big ) + \cdots
 & |\tau| \le  2 \\[.2cm]
 1, &   |\tau| \ge  2.
 \end{cases}
 \end{equation}
 Furthermore, in both cases the terms not shown are a series in  $N^{-2}$, starting at order $N^{-4}$.
\end{prop}

\begin{proof}
In all cases except $\beta = 4$ with $|\tau| \le 1$, the stated expansions follow from (\ref{S6q}) and (\ref{S7q}) together with the large $z$ asymptotic expansion of the digamma function \cite{WiDigamma}
\begin{equation}\label{W1}
\psi(z) \sim \log z - {1 \over 2 z} - {1 \over 12 z^2} + \cdots.
 \end{equation}
 In the case $\beta = 4$ with $|\tau| \le 1$, we must first rewrite
 $ \psi(-N+|k|+1/2)$ using the relection formula for the digamma function,
 \begin{equation}\label{W2}
  \psi(-N+|k|+1/2) = \psi(N+1/2 - |k|),
  \end{equation} 
before making use of (\ref{W1}).

To deduce that all the terms in the large $N$ asymptotic expansion are in powers of $N^{-2}$, we first eliminate the shift by $1/2$ in the arguments of the digamma functions in (\ref{S6q}) and (\ref{S7q}) by using the identity \cite{WiHarmonic}
\begin{equation}\label{W3}
\Psi(n+1/2) = - \gamma - 2 \log 2 + 2H_{2n} - H_{n},
\end{equation} 
where $\gamma$ denotes Euler's constant and $H_n$ denotes the harmonic numbers. Use of the large $n$ asymptotic expansion \cite{WiHarmonic} 
 \begin{equation}\label{W4}
 H_n \sim \log n + \gamma + {1 \over 2n} - \sum_{k=1}^\infty {B_{2k} \over 2 k n^{2k}}
 \end{equation} 
 then establishes the result.
    \end{proof}

    According to Proposition \ref{P3.1}, for $\beta = 1$ and 4 at least the bulk scaled structure function has the large $N$ asymptotic expansion
   \begin{equation}\label{X5} 
   \tilde{S}_{N,\beta}(\tau) \sim
  \tilde{S}_{0,\infty,\beta}(\tau) + {1 \over N^2}  \tilde{S}_{1,\infty,\beta}(\tau) + {1 \over N^4}  \tilde{S}_{2,\infty,\beta}(\tau) + \cdots 
  \end{equation} 
  where the higher order terms are a series in $1/N^2$ starting at order $1/N^6$. From the explicit expansions (\ref{S8q}) and (\ref{S9q}) (further extended to order $1/N^4$), we observe the differential relations
  for $\beta =1,4$
  (where it has been assumed that $\tau > 0$, which then removes the absolute value signs in the expressions of Proposition \ref{P3.1})
  \begin{align}\label{X6}  
& \tilde{S}_{1,\infty,\beta}(\tau) = c_\beta 
 \tau^2 {d^2 \over d \tau^2} \tilde{S}_{0,\infty,\beta}(\tau), \nonumber \\
& \tilde{S}_{2,\infty,\beta}(\tau) = d_\beta
 \bigg ( \tau^4 {d^4 \over d \tau^4} \tilde{S}_{0,\infty,\beta}(\tau) + 8
\tau^3 {d^3 \over d \tau^3} \tilde{S}_{0,\infty,\beta}(\tau) + 12
\tau^2 {d^2 \over d \tau^2} \tilde{S}_{0,\infty,\beta}(\tau) \bigg ),
  \end{align} 
  with $c_1 = -{1 \over 6}, c_4 = - {1 \over 24}$, $d_1 =  {7 \over 360}, d_4 = {7 \over 5760}$. 

  For small $\tau$ and general $\beta > 0$ we have available from \cite[Corollary 4.1]{WF15} the small $|\tau|$ expansion of the terms on the RHS of (\ref{X5}) up to and including ${\rm O}(\tau^{11})$ for $ \tilde{S}_{0,\infty,\beta}(\tau)$
  (the final order term was obtained in the subsequent work \cite{Fo21}), to order $\tau^8$ for $ \tilde{S}_{1,\infty,\beta}(\tau)$, and to order $\tau^6$ for $ \tilde{S}_{2,\infty,\beta}(\tau)$. For present purposes we will be content with presenting these expansions to smaller order than what is available. For this, we define the polynomials
  \begin{align*}
  &p_2(\kappa) = 1 - {11 \kappa \over 6} + \kappa^2, \quad p_4(\kappa) = 1 - {91 \kappa \over 30} + {62 \kappa^2 \over 15} - {91 \kappa^3 \over 30} + \kappa^4, \\
  &q_2(\kappa) = 1 - {3 \kappa \over 2} + \kappa^2, \quad q_4(\kappa) = 1 - {37 \kappa \over 15} + {13 \kappa^2 \over 4} - {37 \kappa^3 \over 15} + \kappa^4, \\
  &r_2(\kappa) = 1 + {15 \kappa \over 8} + \kappa^2, \quad r_4(\kappa) = 1 + {31 \kappa \over 42} - {116 \kappa^2 \over 42} + {31 \kappa^3 \over 42} + \kappa^4, \\
  & p_6(\kappa) = 1 - {1607 \kappa \over 420} + {2011 \kappa^2 \over 280} - {911 \kappa^3 \over 105} + {2011 \kappa^4 \over 280}
  - {1607 \kappa^5 \over 420}
  \kappa^6, \\
  & q_6(\kappa) = 1 - {263 \kappa \over 84} + {1697 \kappa^2 \over 315} - {6337 \kappa^3 \over 1008} + {1697 \kappa^4 \over 315}
  - {263 \kappa^5 \over 84}
  \kappa^6.
\end{align*}
Here $\kappa := \beta/2$. In terms of these polynomials
one has from \cite{WF15} the small $\tau$ expansions
\begin{gather}
 \begin{multlined}\label{2.4X}
\tilde{S}_{0,\infty,\beta}(\tau) = 
{1 \over \kappa} \tau + {(\kappa - 1) \over \kappa^2} \tau^2 + {(\kappa - 1)^2 \over \kappa^3} \tau^3 + 
{(\kappa - 1) \over \kappa^4} p_2(\kappa) \tau^4 + {(\kappa - 1)^2 \over \kappa^5} q_2(\kappa) \tau^5 \\
+ {(\kappa - 1) \over \kappa^6} p_4(\kappa) \tau^6 + {(\kappa - 1)^2 \over \kappa^7} q_4(\kappa) \tau^7 + {(\kappa - 1) \over \kappa^8} p_4(\kappa) \tau^8 + {\rm O}(\tau^9),
\end{multlined}\\
\begin{multlined} \label{2.4Y}
\tilde{S}_{1,\infty,\beta}(\tau) =
-{(\kappa - 1) \over 6 \kappa^3} \tau^2 - 
     {(\kappa - 1)^2 \over 2 \kappa^4} \tau^3 -  {(\kappa - 1) \over  \kappa^5} p_2(\kappa) \tau^4 - {5(\kappa - 1)^2 \over 3 \kappa^6} q_2(\kappa) \tau^5 \\
     -  {5(\kappa - 1) \over  2\kappa^7} p_4(\kappa) \tau^6 + {\rm O}(\tau^7),
 \end{multlined}\\    
\begin{multlined} \label{2.4Z}
\tilde{S}_{2,\infty,\beta}(\tau) =
{(\kappa - 1) \over 30 \kappa^5} (1 + \kappa + \kappa^2) \tau^2 + 
     {2 (\kappa - 1)^2 \over 15 \kappa^6} r_2(\kappa) \tau^3 +  {7(\kappa - 1) \over  20 \kappa^7} r_4(\kappa) \tau^4  +  {\rm O}(\tau^5).
   \end{multlined}
    \end{gather}
These general $\beta$ expansions (\ref{2.4X}) and (\ref{2.4Y}) are  seen to be consistent with the first differential relation in (\ref{X6}) with $c_\beta = - {1 \over 12 \kappa}$, while (\ref{2.4X}) and (\ref{2.4Z}) are consistent with the second differential relation with $d_\beta = (\kappa^3 - 1)/(720 \kappa^3 (\kappa - 1))$.

    \begin{remark} ${}$ \\
    1.~The limiting functional forms of $ \tilde{S}_{N,\beta}(\tau)$ are classical results in random matrix theory (see \cite[Eq.~(5.6) for $\beta = 1$]{Dy62a} and \cite[equivalent to Eq.~(7.2.46) for $\beta = 4$]{Me04}). \\
    2.~By inspection the expansions (\ref{2.4X})--(\ref{2.4Z}) are unchanged (up to a minus sign) by the mappings $\kappa \mapsto 1/\kappa$ and $\tau \mapsto - \tau/\kappa$. This 
    is consistent with the known functional equation \cite[Prop.~4.7]{WF15}. Also, previously the polynomials of $\kappa$
 appearing in (\ref{2.4X}) have been observed to have all their zeros on the unit circle in the complex $\kappa$ plane, and which display an interlacing property \cite{WF14}, \cite{FS23}. This is similarly true of the polynomials in  (\ref{2.4Y}) and (\ref{2.4Z}).
 \end{remark}

\section{Bulk scaling expansion of the correlations for even $\beta$}\label{S3}
\subsection{Hypergeometric functions based on Jack polynomials}
In the cases $\beta = 1,2$ and 4 the general $k$-point correlation function for the circular $\beta$ ensemble are fully determined by a certain scalar kernel function (\ref{3.1}) for $\beta = 2$, and matrix kernel function (\ref{K14}) for $\beta = 1$ and 4. Generally, from (\ref{1.1A}) and the $\beta$ version of (\ref{1.2}), the cases of $\beta$ even have the special property that the $k$-point correlation function is a Laurent polynomial in $\{e^{i \theta_j} \}_{j=1,\dots,k}$. In the work \cite{Fo92j} this polynomial was identified in terms of a (then) recently introduced class of generalised hypergeometric functions based on Jack polynomials \cite{Ka93}. To define these hypergeometric functions we must then first introduce Jack polynomials  \cite[Ch.~VI.6]{Ma95}, \cite[Ch.~12]{Fo10}, \cite[Ch.~7]{KK09}.

 Jack polynomials $P_\kappa^{(\alpha)}(\mathbf x)$ are symmetric polynomials of $N$ variables
 $\mathbf x = (x_1,\dots,x_N)$. They are labelled by a partition $\kappa :=
 (\kappa_1,\dots,\kappa_N)$ with parts which are non-negative
 integers ordered as $\kappa_1 \ge \kappa_2 \ge \cdots \ge \kappa_N$,
 and dependent on a parameter $\alpha > 0$. 
They have the triangular structure with respect to the monomial basis  $\{ m_\mu(\mathbf x) \}$ 
\begin{equation}\label{4.0a}
P_\kappa^{(\alpha)}(\mathbf x) = m_\kappa(\mathbf x) +
\sum_{\mu < \kappa} c_{\kappa, \mu}^{(\alpha)} 
m_\mu (\mathbf x).
\end{equation}
In (\ref{4.0a})  the 
notation $\mu < \kappa$ denotes the partial order on partitions, defined by the requirement that 
 for each $s=1,\dots,\ell(\kappa)$ (where $\ell(\kappa)$  denotes the number of nonzero parts of $\kappa$)
 we have
$\sum_{i=1}^s \mu_i \le \sum_{i=1}^s \kappa_i$. The Jack polynomial $P_\kappa^{(\alpha)}(\mathbf x)$ can be specified as they unique
 polynomial eigenfunctions of the
differential operator
\begin{equation}\label{4.0}
\sum_{j=1}^N
x_j^2 {\partial^2 \over \partial x_j^2}  
+ {2 \over \alpha}  \sum_{1 \le j < k \le N}{1 \over x_j -
x_k}
\Big (x_j^2 {\partial \over \partial x_j} - x_k^2 {\partial
\over
\partial x_k}
\Big )
\end{equation}
with the structure (\ref{4.0a}).

The special case $\alpha =2$ of the Jack polynomials correspond (up to normalisation) to the zonal polynomials of mathematical statistics; see e.g.~\cite[Ch.~7]{Mu82}. In this case, generalised hypergeometric functions 
of the type of relevance to the circular $\beta$ ensemble were introduced in a work of Herz \cite{He55}, albeit entirely in the context of matrix integrals, and without reference to zonal polynomials. A series definition involving zonal polynomials came in the later work of Constantine \cite{Co63}. Generalising this, Yan \cite{Ya92} and Kaneko
\cite{Ka93} introduced the
family of hypergeometric functions based on Jack polynomials as 
\begin{equation}\label{3.40}
  {\vphantom{F}}_p^{\mathstrut} F_q^{(\alpha)}(a_1,\dots,a_p;b_1,\dots,b_q; \mathbf x):=\sum_\kappa  {\alpha^{| \kappa |} \over h_\kappa'} 
 \frac{[a_1]^{(\alpha)}_\kappa\dots [a_p]^{(\alpha)}_\kappa }{[b_1]^{(\alpha)}_\kappa
\dots [b_q]^{(\alpha)}_\kappa} 
P_\kappa^{(\alpha)}(\mathbf x). 
\end{equation}
Here,
with $\mathbf x = \{x_i\}_{i=1}^m$, the sum is over all partitions $\kappa_1 \ge \kappa_2 \ge \cdots \ge \kappa_m \ge 0$, and conventionally the sum is
performed in order of increasing 
$|\kappa|$. Also use has been made of
  the generalised Pochhammer symbol
  \begin{equation}\label{4.3m}
\quad[u]_\kappa^{(\alpha)} := \prod_{l=1}^N {\Gamma(u - (j-1)/\alpha + \kappa_l) \over \Gamma(u - (j - 1)/\alpha)};
\end{equation}
see e.g.~\cite[Eq.~(12.46)]{Fo10}. For the precise definition of $h_\kappa'$ (which also depends on $\alpha$) we reference e.g.~\cite[Eq.~(3.4)]{Fo24}. Recent works making essential use of this class of hypergeometric function in contexts relating to random matrix theory include \cite{FW21,Fo22,No24,Wi24,LWJZ25}.

\subsection{Evenness of the correlation functions in $N$}
Deduced in \cite{Fo92j}, and conveniently summarised in \cite[Prop.~13.2.1]{Fo10}, one has that the $n$-point correlation function for circular $\beta$ ensemble with even $\beta$ 
 in bulk scaling variables
 $$
 \Big ( {2 \pi \over N} \Big )^n
 \rho_{(n)}\left (2 \pi r_1/N, \ldots, 2 \pi r_n/N \right) =: 
\tilde{\rho}_{(n),\beta}\left(r_1, \ldots, r_n\right) 
 $$
 has the evaluation in terms of the generalized hypergeometric function 
\begin{multline}\label{4.5}
\tilde{\rho}_{(n),\beta}\left(r_1, \ldots, r_n\right)= \frac{(N-n+1)_n((\beta/2)!)^N}{N^n\Gamma\left(\frac{\beta N}{2}+1\right)}\prod_{1\leq j<k\leq n}|e^{2\pi ir_k/N}-e^{2\pi ir_j/N}|^\beta \\
     \times M_{N-n}(n \beta / 2, n \beta / 2, \beta / 2) \prod_{k=2}^{n}e^{2\pi i\beta(r_k-r_1)(N-n)/N} {\vphantom{F}}_2^{\mathstrut} F_1^{(\beta/2)} \left(-N+n, n ; 2 n ; 1-t_1, \ldots, 1-t_{(n-1) \beta}\right),
\end{multline}
where 
\begin{align*}
    t_k:=e^{-2 \pi i\left(r_j-r_1\right) / L}, \quad k=1+(j-2) \beta, \ldots,(j-1) \beta, \quad (j=2, \ldots, n),
\end{align*}
and $M_{N-n}$ is the Morris integral (see e.g.~\cite[Eq.~(4.4)]{Fo10})
\begin{align}
    \begin{aligned}
M_N(a, b, \lambda) & :=\int_{-1 / 2}^{1 / 2} d \theta_1 \cdots \int_{-1 / 2}^{1 / 2} d \theta_N \prod_{l=1}^N e^{\pi i \theta_l(a-b)}\left|1+e^{2 \pi i \theta_l}\right|^{a+b} \prod_{1 \leq j<k \leq N}\left|e^{2 \pi i \theta_k}-e^{2 \pi i \theta_j}\right|^{2 \lambda} \\
& =\prod_{j=0}^{N-1} \frac{\Gamma(\lambda j+a+b+1) \Gamma(\lambda(j+1)+1)}{\Gamma(\lambda j+a+1) \Gamma(\lambda j+b+1) \Gamma(1+\lambda)}.
\end{aligned}
\end{align}
We note that the RHS of (\ref{4.5}) gives meaning to $\tilde{\rho}_n$ for continuous $N$.
In fact, as already demonstrated above for $\beta = 1,2$ and 4 as properties of the underlying correlation kernels, this can be shown to be even in $N$.

\begin{prop}\label{P4.1}
Fix $n$ independent of $N$ and such that $2n < N$. We have that (\ref{4.5}) is an even function of $N$.
\end{prop}

\begin{proof}
 The first observation is that by using the transformation formula (see e.g.~\cite[Prop.~13.1.6]{Fo10})
\begin{align}
    {\vphantom{F}}_2^{\mathstrut} F_1^{(\alpha)}\left(a, b ; c ; t_1, \ldots, t_m\right)=\prod_{j=1}^m\left(1-t_j\right)^{-b}
    {\vphantom{F}}_2^{\mathstrut} F_1^{(\alpha)}
    \left(c-a, b ; c ;-\frac{t_1}{1-t_1}, \ldots,-\frac{t_m}{1-t_m}\right),
\end{align}
it follows that the second line of (\ref{4.5}) is even in $N$.
It is immediate that the product $\prod_{1\leq j<k\leq n}|e^{2\pi ir_k/N}-e^{2\pi ir_j/N}|^\beta$ is also invariant under negation of $N$. The remaining factor, after substitution of the Morris integral formula, is 
\begin{align}
    \frac{(N-n+1)_n((\beta/2)!)^n}{N^n\Gamma\left(\frac{\beta N}{2}+1\right)}\prod_{j=0}^{N-n-1}\frac{\Gamma\left(\frac{\beta}{2}(j+2n)+1\right)\Gamma\left(\frac{\beta}{2}(j+1)+1\right)}{\Gamma\left(\frac{\beta}{2}(j+n)+1\right)^2}.
\end{align}
Assume that $N>2n$ and denote $\beta/2=\kappa$
as used in the paragraph below (\ref{X6}), then the cancellation in the product of the gamma functions yields
\begin{align}
    \frac{(N-n+1)_n(\kappa!)^n}{N^n\Gamma(n\kappa+1)}\prod_{k=1}^{n-1}\frac{\Gamma(\kappa(N+k)+1)\Gamma(\kappa k+1)}{\Gamma(\kappa(N-n+k)+1)\Gamma(\kappa(n+k)+1)},
\end{align}
in which the $N$-dependent part is
\begin{align}
\frac{(N-n+1)_n}{N^n}\prod_{k=1}^{n-1}\frac{\Gamma(\kappa(N+k)+1)}{\Gamma(\kappa(N-n+k)+1)}=(\kappa N)^{-n+1}\prod_{k=1}^{n-1}\frac{\Gamma(\kappa(N+k)+1)}{\Gamma(\kappa(N-n+k))} .
\end{align}
This can then be further simplified to give
\begin{align}
    (\kappa N)^{-n+1}\prod_{k=1}^{n-1}\frac{\Gamma(\kappa(N+k)+1)}{\Gamma(\kappa(N-k))}=\prod_{k=1}^{n-1}\prod_{l=1}^{k\kappa}(\kappa^2N^2-l^2),
\end{align}
which is thus also even in $N$, and so all terms in (\ref{4.5}) are even in $N$.
\end{proof}

It follows from Proposition \ref{P4.1} that if, for $\beta$ even, the existence of an asymptotic expansion in inverse powers of $N$ can be established, then in fact the expansion is in inverse powers of $N^{-2}$. Working relevant to establishing that the generalised hypergeometric function in (\ref{4.5}) permits an expansion in $1/N$ up to order $1/N^2$ at least has been given in the recent work \cite{Wi24} (strictly speaking the dependence on $N$ in the case of the particular ${\vphantom{F}}_2^{\mathstrut} F_1^{(\alpha)}$ appearing in \cite{Wi24} is different to its appearance in (\ref{4.5}) so this is not immediate). However, restricting attention to the two-point correlation, one has available a particular $\beta$-dimensional integral form of the required
generalised hypergeometric function from which the existence of a $1/N$ expansion is easy to establish. This integral form reads
\cite{Fo94j}, \cite[Prop.~13.2.2, after deleting a spurious prefactor $(\beta N/2)!$]{Fo10},
\begin{align} 
\tilde{\rho}_{(2),\beta}\left(r_1, r_2\right)= & \frac{(N-1)((\beta / 2)!)^{N}}{N(\beta N / 2)!} \frac{M_{N-2}(\beta , \beta , \beta / 2)}{S_\beta(1-2 / \beta, 1-2 / \beta, 2 / \beta)} \nonumber \\
& \times\left(2 \sin \pi\left(r_1-r_2\right) / N\right)^\beta e^{-\pi i \beta \left(r_1-r_2\right)(N-2) / N} \int_{[0,1]^\beta} d u_1 \cdots d u_\beta \nonumber \\
& \times \prod_{j=1}^\beta\left(1-\left(1-e^{2 \pi i\left(r_1-r_2\right) / N}\right) u_j\right)^{N-2} u_j^{-1+2 / \beta}\left(1-u_j\right)^{-1+2 / \beta} \prod_{j<k}\left|u_k-u_j\right|^{4 / \beta}. \label{4.12b}
\end{align}
Here $S_n(a,b,c)$ is the Selberg integral (see e.g.~\cite[\S 4.1]{Fo10}), which for the stated parameters is such that the multiple integral is normalised to unity for $r_1-r_2=0$.

\begin{prop}\label{P4.2}
    The RHS of (\ref{4.12b}) admits an asymptotic expansion in $1/N$. Combining this with Proposition \ref{P4.1} gives the $1/N^2$ expansion
    \begin{equation}
    \tilde{\rho}_{(2),\beta}\left(x, 0\right)
    \sim  \rho_{(2),0,\beta}^{\rm bulk}(x,0) +
    \sum_{l=1}^\infty {1 \over N^{2l} } \rho_{(2),l,\beta}^{\rm bulk}(x,0),\label{exp-rho}
\end{equation}    
for some $\{ \rho_{(2),l,\beta}^{\rm bulk}(x,0) \}_{l=1,2,\dots}$,   and where
\begin{multline}
\rho_{(2),0,\beta}^{\rm bulk} (x,0)  = 
 (\beta / 2)^\beta {((\beta /2)!)^3 \over \beta! (3 \beta /2)!} 
{e^{- \pi i \beta  x} (2 \pi  x)^\beta \over
 S_\beta(-1+ 2/\beta,-1+ 2/\beta,2/\beta)} \\
\times \int_{[0,1]^\beta} du_1 \cdots du_\beta
 \prod_{j=1}^\beta
e^{2 \pi i x u_j} u_j^{-1 + 2/\beta} (1 - u_j)^{-1 + 2/\beta}
\prod_{j < k} |u_k - u_j|^{4/\beta}.  \label{15.eev1}
\end{multline}
\end{prop}

\begin{proof}
With $A_N(x) := 1 + \sum_{l=1}^\infty {(2\pi ix)^l \over (l+1)! N^l}$, we write
\begin{multline}
\prod_{j=1}^\beta\left(1-\left(1-e^{2 \pi i x / N}\right) u_j\right)^N \\
= \prod_{j=1}^\beta e^{N \log  (1 + {2 \pi i x u_j \over N} A_N(x)  )} 
= \prod_{j=1}^\beta  e^{2 \pi i x u_j (1+\sum_{l=1}^\infty {(2 \pi i x )^l \over (l+1)! N^l})}
e^{2 \pi i x u_j \sum_{p=2}^\infty {(-1)^{p-1} \over p N^{p+1}} (A_N(x))^p}
\end{multline}
Since for fixed $x$ this expression admits a $1/N$ expansion, uniform in the integration variables on their domain of support,
it follows immediately that the multi-dimensional integral similarly admits a $1/N$ expansion. In addition, the working of the proof of Proposition \ref{P4.1} shows that all the prefactors to the multi-dimensional integral admit $1/N$ expansions, which then as said is lifted to a $1/N^2$ expansion according to the result of Proposition \ref{P4.1}. The explicit form of the limit
(\ref{15.eev1}), already known from \cite{Fo94j}, can be read off by keeping track of the leading order terms throughout.
\end{proof}

\subsection{Functional form of the $1/N^2$ correction for the two-point function}

Starting with (\ref{4.12b}), it is possible to deduce the generalisation of the formulas
(\ref{5.1}) and  (\ref{5.1B})  linking the first correction $\rho_{(2),1,\beta}^{\rm bulk}(x,0)$ in (\ref{exp-rho}) to the limit
$\rho_{(2),0,\beta}^{\rm bulk}(x,0)$ for all
even $\beta$.

\begin{prop}
With reference to the asymptotic expansion (\ref{exp-rho}), for all even $\beta$ we have
\begin{equation}\label{4.21}
    \rho_{(2),1,\beta}^{\rm bulk}(x,0)=-\frac{1}{6\beta}\frac{d^2}{dx^2}\bigg(x^2\rho_{(2),0,\beta}^{\rm bulk} (x,0)\bigg).
\end{equation}
    \end{prop}

\begin{proof}
In order to obtain the functional form of the $1/N^2$ correction, we start from the calculation of the $1/N$ expansion of the integral 
\begin{align}\label{int1}
    \int_{[0,1]^\beta} d u_1 \cdots d u_\beta\prod_{j=1}^\beta\left(1-\left(1-e^{2 \pi i\left(r_1-r_2\right) / N}\right) u_j\right)^{N-2} u_j^{-1+2 / \beta}\left(1-u_j\right)^{-1+2 / \beta} \prod_{j<k}\left|u_k-u_j\right|^{4 / \beta}.
\end{align}
Introduce the notation
\begin{align*}
    \mathcal{I}[f(u_1,\dots,u_\beta)]=\int_{[0,1]^\beta} d u_1 \cdots d u_\beta \prod_{j=1}^{\beta}e^{i \theta u_j} u_j^{-1+2 / \beta}\left(1-u_j\right)^{-1+2 / \beta} \prod_{j<k}\left|u_k-u_j\right|^{4 / \beta}f(u_1,\dots,u_\beta).
\end{align*}
Let $\theta=2\pi x$, then by substituting the expansion up to terms O$(N^{-3})$
\begin{align*}
    \bigg(1-&\left(1-e^{ i\theta / N}\right) u_j\bigg)^{N-2}\sim e^{i \theta u_j}+\frac{\theta u_j e^{i \theta 
   u_j} (\theta  u_j-\theta -4 i)}{2 N}\\
   &+\frac{e^{i \theta u_j}}{N^2}\bigg(\frac{1}{2}\bigg(\frac{\theta^2}{2}(u_j^2-u_j)-2 i\theta u_j\bigg)^2-\theta^2(u_j^2-u_j)+\frac{i\theta^3}{6}u_j(u_j-1)+\frac{i\theta^3}{3}u_j^2(1-u_j)\bigg),
\end{align*}
we see that \eqref{int1} can be expanded 
up to terms O$(N^{-3})$ as
\begin{align}\label{expdI}
    \mathcal{I}[1]+& \mathcal{I}\bigg[\sum_{j=1}^{\beta} \theta u_j (\theta  u_j-\theta -4 i)\bigg]\frac{1}{2N}+\frac{1}{N^2}\bigg(\frac{\theta^4}{8}\mathcal{I}\bigg[\sum_{j=1}^{\beta}u_j^2(1-u_j)^2\bigg]+i\theta^3\mathcal{I}\bigg[\sum_{j=1}^{\beta}u_j^2(1-u_j)\bigg]\nonumber \\
    &-2\theta^2\mathcal{I}\bigg[\sum_{j=1}^{\beta}u_j^2\bigg]+\theta^2\bigg(1-\frac{i\theta}{6}\bigg)\mathcal{I}\bigg[\sum_{j=1}^{\beta}u_j(1-u_j)\bigg]+\frac{i\theta^3}{3}\mathcal{I}\bigg[\sum_{j=1}^{\beta}u_j^2(1-u_j)\bigg]\bigg)\nonumber \\
    &+\frac{1}{8N^2}\bigg(\theta^4\mathcal{I}\bigg[\sum_{j\neq k}u_j(1-u_j)u_k(1-u_k)\bigg]+8i\theta^3\mathcal{I}\bigg[\sum_{j\neq k}u_ju_k(1-u_k)\bigg]-16\theta^2\mathcal{I}\bigg[\sum_{j\neq k}u_ju_k\bigg]\bigg).
\end{align}

The integrals in the above expansion can be evaluated by performing integration by parts on each variable $u_j$ and then taking the sum. 
This general strategy was introduced into the wider theory of Selberg integrals by Aomoto \cite{Ao87}; see \cite[\S 4.6]{Fo10} for more on this. With
\begin{align*}
    f(u):=\prod_{j=1}^{\beta}e^{i \theta u_j} u_j^{-1+2 / \beta}\left(1-u_j\right)^{-1+2 / \beta} \prod_{j<k}\left|u_k-u_j\right|^{4 / \beta},
\end{align*}
taking the partial derivative on each side of $u_k^m(1-u_k)^nf(u)$ gives
\begin{align*}
    \partial_{u_k}&(u_k^m(1-u_k)^nf(u))\\
    &=\bigg(i\theta +\bigg(\frac{2}{\beta}+m-1\bigg)\frac{1}{u_k}-\bigg(\frac{2}{\beta}+n-1\bigg)\frac{1}{1-u_k}+\sum_{\substack{l=1\\ l\neq k}}^{\beta}\frac{1}{u_k-u_l}\bigg)u_k^m(1-u_k)^nf(u).
\end{align*}
Starting from the simplest case with $m=n=1$, using integration by parts, this shows
\begin{align*}
    i\theta \mathcal{I}\bigg[\sum_{k=1}^{\beta}u_k(1-u_k)\bigg]&=-\frac{2}{\beta}\mathcal{I}\bigg[\sum_{k=1}^{\beta}(1-u_k)\bigg]+\frac{2}{\beta}\mathcal{I}\bigg[\sum_{k=1}^{\beta}u_k\bigg]-\frac{4}{\beta}\mathcal{I}\bigg[\sum_{l\neq k}\frac{u_k(1-u_k)}{u_k-u_l}\bigg]\\
    &=-2\beta\mathcal{I}[1]+4\mathcal{I}\bigg[\sum_{k=1}^{\beta}u_k\bigg]=-2\beta\mathcal{I}[1]-4i\partial_\theta\mathcal{I}[1],
\end{align*}
where we have used the formula
\begin{align*}
    \sum_{k=1}^{\beta}\sum_{\substack{l=1\\l\neq k}}^{\beta}\frac{u_k(1-u_k)}{u_k-u_l}=\frac{\beta(\beta-1)}{2}-(\beta-1)\sum_{k=1}^\beta u_k.
\end{align*}
Similar calculations show 
\begin{align*}
    -i\theta\mathcal{I}\bigg[\sum_{j=1}^{\beta}u_j^2(1-u_j)\bigg]=&\bigg(\frac{2}{\beta}+1\bigg)\mathcal{I}\bigg[\sum_{j=1}^{\beta}u_j(1-u_j)\bigg]-\frac{2}{\beta}\mathcal{I}\bigg[\sum_{j=1}^{\beta}u_j^2\bigg]+\frac{4}{\beta}\mathcal{I}\bigg[\sum_{j\neq l}\frac{u_j^2(1-u_j)}{u_j-u_l}\bigg] \\
    =&\bigg(5-\frac{2}{\beta}\bigg)\mathcal{I}\bigg[\sum_{j=1}^{\beta}u_j(1-u_j)\bigg]-\frac{2}{\beta}\mathcal{I}\bigg[\bigg(\sum_{j=1}^{\beta}u_j\bigg)^2\bigg],\\
    -i\theta\mathcal{I}\bigg[\sum_{j=1}^{\beta}u_j^2(1-u_j)^2\bigg]=&
    \bigg(5-\frac{2}{\beta}\bigg)\mathcal{I}\bigg[\sum_{j=1}^{\beta}u_j(1-u_j)\bigg]-6\mathcal{I}\bigg[\sum_{j=1}^{\beta}u_j^2(1-u_j)\bigg]-\frac{4}{\beta}\mathcal{I}\bigg[\sum_{j\neq l}u_lu_j(1-u_j)\bigg],
\end{align*}
and
\begin{small}
    \begin{align*}
    &- i\theta \mathcal{I}\left[\sum_{j \neq k} u_j\left(1-u_j\right) u_k\left(1-u_k\right)\right]= \frac{2}{\beta}(\beta-1)^2\mathcal{I}\left[\sum_{k=1}^\beta u_k\left(1-u_k\right)\right]+\left(\frac{4}{\beta}-4\right) \mathcal{I}\left[\sum_{j \neq k} u_j u_k\left(1-u_k\right)\right],\\
    &- i\theta \mathcal{I}\left[\sum_{j \neq k} u_j u_k\left(1-u_k\right)\right]
    =2 \beta \mathcal{I}\left[\sum_{j=1}^{\beta} u_j\right]+\frac{2}{\beta}(1-2 \beta)\left(\mathcal{I}\left[\sum_{j=1}^{\beta} u_j\left(1-u_j\right)\right]+\mathcal{I}\left[\left(\sum_{j=1}^{\beta} u_j\right)^2\right]\right).
\end{align*}
\end{small}

 Substituting the above equations in \eqref{expdI}, we see that \eqref{int1} can be written up to terms ${\rm O}(N^{-3})$
\begin{align}\label{Az}
    \mathcal{I}[1]-i\theta\beta\mathcal{I}[1]\frac{1}{N}+\frac{\theta}{12 \beta }  ((\beta\theta+2 i (\beta +2)) (\beta 
   \mathcal{I}[1]+2 i \partial_{\theta} \mathcal{I}[1])-6\beta^3\theta\mathcal{I}[1]-2 \theta\partial_{\theta}^2 \mathcal{I}[1]  )\frac{1}{N^2},
\end{align}
where we have used the simple relations
\begin{align*}
    \mathcal{I}\left[\sum_{j=1}^{\beta} u_j\right]=-i\partial_\theta\mathcal{I}[1],\quad \mathcal{I}\left[\bigg(\sum_{j=1}^{\beta} u_j\bigg)^2\right]=-\partial_\theta^2\mathcal{I}[1].
\end{align*}
Now we look at the prefactors in \eqref{4.12b}. After simplification, we have
\begin{align*}
    {((\beta /2)!)^3 \over \beta! (3 \beta /2)!S_\beta\left(1 - 2/\beta,\, 1 - 2/\beta,\, 2/\beta\right)}
\frac{
  (N - 1)\,\Gamma\left(\beta(N + 1)/2 + 1\right)
}{
  N\,\Gamma\left(\beta(N - 1)/2 + 1\right)
}
\left( 2 \sin\left( \frac{\theta}{2N} \right) \right)^{\beta}
e^{- i \beta \theta \frac{(N - 2)}{2N}}.
\end{align*}
For large $N$
\begin{align*}
   &\frac{
  (N - 1)\,\Gamma\left(\beta(N + 1)/2 + 1\right)
}{
  N\,\Gamma\left(\beta(N - 1)/2 + 1\right)
}=\bigg(\frac{\beta}{2}\bigg)^\beta \bigg(1+\frac{-\beta ^2-3 \beta -2}{6 \beta 
   N^2}+{\rm O}\left(N^{-3}\right)\bigg),\\
   &2 \sin\left( \frac{\theta}{2N}\right)=\theta-\frac{\theta^3}{24 N^2}+{\rm O}\left(N^{-4}\right),\quad e^{- i \beta \theta \frac{(N - 2)}{2N}}=e^{-\frac{1}{2} i \beta \theta} \left(1+\frac{i \beta
   \theta}{N}-\frac{\beta^2 \theta^2}{2 N^2}+{\rm O}(N^{-3})\right).
\end{align*}
Multiplying the so implied $1/N$ expansion of the prefactors with the $1/N$ expansion (\ref{Az}) show that the first correction 
in \eqref{exp-rho} (i.e.~the coefficient of $1/N^2)$ is given by
\begin{align}\label{rho21b}
    \rho_{(2),1,\beta}^{\rm bulk}(x,0)=&\mathcal{A}_\beta \frac{x^\beta e^{-\pi i\beta x}}{6\beta}\bigg((\beta^2(\pi x+i)^2+\beta(4\pi ix-3)-2)\mathcal{I}[1]\notag\\
    &-2 x(2+\beta-\pi i\beta x)\partial_x \mathcal{I}[1]- x^2\partial_x^2\mathcal{I}[1]\bigg),
\end{align}
where
$$
    \mathcal{A}_\beta= {2^{\beta}\pi^\beta(\beta / 2)^\beta((\beta /2)!)^3 \over \beta! (3 \beta /2)!S_\beta\left(1 - 2/\beta,\, 1 - 2/\beta,\, 2/\beta\right)},\quad
    \mathcal{I}[1]=\frac{1}{\mathcal{A}_\beta}x^{-\beta}e^{\pi i\beta x}\rho_{(2),0,\beta}^{\rm bulk} (x,0).
$$
By substituting the later equation into \eqref{rho21b}, the stated relation (\ref{4.21}) follows.
\end{proof}

\begin{remark} ${}$
1.~In Appendix A, it is shown that the above calculations can be generalized to a recurrence relation among 
\begin{align}\label{AB1}
    \mathcal{I}^{(m)}(a_1,a_2,\dots,a_m):=\mathcal{I}\bigg[\sum_{j_1\neq j_2\neq\dots\neq j_m}u_{j_1}^{a_1}\dots u_{j_m}^{a_m} \bigg],\quad m\in\mathbb{Z}_{\geq 0},
\end{align}
which can be used to show that
the higher order corrections in \eqref{exp-rho}
are related to a finite sum over derivatives of the limit, where the coefficients in the sum are polynomials. \\
2.~As well as conjecturing that (\ref{4.21}), proved now for even $\beta$ and $\beta =1$, holds true for general $\beta > 0$, we can conjecture too that so does the $\beta$ generalisation of (\ref{5.1b}), (\ref{5.1e})
and (\ref{5.1g}),
\begin{equation}\label{5.1Ba}
 \mathcal P_{1,\beta}^{\rm bulk}(s;\xi) =
 -{1 \over 6 \beta} {d^2 \over ds^2}
 \bigg ( s^2 \mathcal P_{0,\beta}^{\rm bulk}(s;\xi) \bigg ),
\end{equation}
with (\ref{4.21}) corresponding to the case $\xi =0$. Further evidence for this is presented in Appendix B, where it is shown that (\ref{5.1Ba}) is valid for the leading $s$ powers at each order $\xi^k$ in the bulk scaling expansion of (\ref{4.3b}). 
\end{remark}

  \subsection*{Acknowledgements}
	This work  was supported
	by the Australian Research Council 
	 Discovery Project grants 
     DP210102887 and DP250102552.  Correspondence with F.~Bornemann relating to this work has been appreciated.

     \appendix
\section*{Appendix A}
\renewcommand{\thesection}{A} 
\setcounter{equation}{0}
\setcounter{prop}{0}
Here a recurrence relation for the quanitities
(\ref{AB1}) will be deduced, which has consequence for the functional form of the higher order corrections in \eqref{exp-rho}.
For this purpose, given an integer $m$, introduce the notations
\begin{align*}
    &[m]=\{1,2,\dots,m\},\quad a^{[m]}=\{a_1,\dots,a_m\},\\
    &{\rm J}^{[m]}=\{(j_1,\dots,j_m)| j_1,j_2,\dots,j_m\in[\beta],j_1\neq j_2\neq\dots\neq j_m\}.
\end{align*}
Denote by $\tau=(\tau_1,\dots,\tau_n)$ a partition of $m$ with length $|\tau|=n$. Associate to $\tau $ a partition $s_\tau$ of $|\tau|$, i.e. $s_\tau=(s_1,\dots,s_k)$ satisfying $\sum_{i=1}^{k}s_i=n $, with each component indicating the number of identical elements in $\tau$. As a set of integers, we adopt the notation $\tau!=\tau_1!\tau_2!\dots\tau_n!$. We have the following recursive identity.
\begin{prop}\label{propA1}
    For any set of integers $a^{[m]}$ with $a_1\geq 2$, the integrals $\mathcal{I}^{(m)}(a_1,a_2,\dots,a_m)$ satisfy the recurrence relation
    \begin{small}
        \begin{align}\label{recur-rl-I}
        \begin{aligned}
            -i\theta&(\mathcal{I}^{(m)}(a_1-1,a^{[m]}\backslash\{a_1\})-\mathcal{I}^{(m)}(a^{[m]}))\\
        =&\frac{4}{\beta}\bigg[\sum_{j=1}^{2}(-1)^j\sum_{k=0}^{\lfloor \frac{a_1-j}{2}\rfloor}\mathcal{I}^{(m+1)}(k,a_1-j-k,a^{[m]}\backslash\{a_1\})-\frac{1}{2}e(a_1-j)\mathcal{I}^{(m+1)}\bigg(\frac{a_1-j}{2},\frac{a_1-j}{2},a^{[m]}\backslash\{a_1\}\bigg)\\
        &+\sum_{k=2}^{m}{\rm sgn}(a_1-a_k)\bigg(\sum_{j=1}^{2} (-1)^{j}\sum_{l={\rm min}\{a_1+1-j,a_k\}}^{\lfloor \frac{a_1+a_k-j}{2}\rfloor}\mathcal{I}^{(m)}(l,a_1+a_k-j-l,a^{[m]}\backslash\{a_1,a_k\})\\
        &-\frac{1}{2}e(a_1+a_k-j)\mathcal{I}^{(m)}\bigg(\frac{a_1+a_k-j}{2},\frac{a_1+a_k-j}{2},a^{[m]}\backslash\{a_1,a_k\}\bigg)\bigg)\bigg]\\
        &+\bigg(\frac{2}{\beta}+a_1-2\bigg)\mathcal{I}^{(m)}(a_1-2,a^{[m]}\backslash\{a_1\})-\bigg(\frac{4}{\beta}+a_1-2\bigg)\mathcal{I}^{(m)}(a_1-1,a^{[m]}\backslash\{a_1\}),
        \end{aligned}
    \end{align}
    \end{small}
    where 
    \begin{align*}
        e(x)=\left\{\begin{array}{cc}
            1, & x\text{ is even}, \\
            0, & x\text{ is odd}.
        \end{array}\right.
    \end{align*}
    The initial conditions $\mathcal{I}^{(m)}(0,\dots,0,1,\dots,1)$ can be determined in terms of $\partial_\theta^m \mathcal{I}[1], m\in \mathbb{Z}_{\geq 0}$ by
    \begin{align*}
        \mathcal{I}^{(m)}(0,a_2,a_3,\dots,a_m)=(\beta-m+1)\mathcal{I}^{(m-1)}(a_2,a_3,\dots,a_m)
    \end{align*}
    and
    \begin{align*}
        (-i)^m\partial_{\theta}^m\mathcal{I}[1]= \mathcal{I}\bigg[\bigg(\sum_{j=1}^\beta u_j\bigg)^m\bigg]=\sum_{\tau\in\mathcal{P}(m)}\frac{m!}{\tau ! s_\tau !}\mathcal{I}^{(|\tau|)}(\tau_1,\dots,\tau_n).
    \end{align*}
\end{prop}

\begin{proof}
        Integration by parts shows
        \begin{align}\label{p1}
            \begin{aligned}
                -i\theta &(\mathcal{I}^{(m)}(a_1-1,a^{[m]}\backslash\{a_1\})-\mathcal{I}^{(m)}(a^{[m]}))\\
            =&\bigg(\frac{2}{\beta}+a_1-2\bigg)\mathcal{I}^{(m)}(a_1-2,a^{[m]}\backslash\{a_1\})-\bigg(\frac{4}{\beta}+a_1-2\bigg)\mathcal{I}^{(m)}(a_1-1,a^{[m]}\backslash\{a_1\})\\
            &+\frac{4}{\beta}\mathcal{I}\bigg[\sum_{{\rm J^{[m]}}}\sum_{\substack{j_{m+1}=1\\j_{m+1}\neq j_1}}^{\beta}\frac{u_{j_1}^{a_1-1}(1-u_{j_1})u_{j_2}^{a_2}\dots u_{j_m}^{a_m}}{u_{j_1}-u_{j_{m+1}}}\bigg].
            \end{aligned}
        \end{align}
        The summation in the last integral can be separated into two sums
        \begin{align*}
            \sum_{{\rm J^{[m+1]}}}\frac{u_{j_1}^{a_1-1}(1-u_{j_1})u_{j_2}^{a_2}\dots u_{j_m}^{a_m}}{u_{j_1}-u_{j_{m+1}}}+\sum_{k=2}^{m}\sum_{{\rm J^{[m]}}}\frac{u_{j_1}^{a_1-1}(1-u_{j_1})u_{j_2}^{a_2}\dots u_{j_m}^{a_m}}{u_{j_1}-u_{j_k}}.
        \end{align*}
        By interchanging the summation indices, the first sum can be written as
        \begin{align*}
            &\sum_{{\rm J}^{[m]\backslash \{1\}}}u_{j_2}^{a_2}\dots u_{j_m}^{a_m}\sum_{\substack{j_1>j_{m+1}\\ j_1,j_{m+1}\neq j_2,\dots,j_m}}\frac{u_{j_1}^{a_1-1}(1-u_{j_1})-u_{j_{m+1}}^{a_1-1}(1-u_{j_{m+1}})}{u_{j_1}-u_{j_{m+1}}}\\
            &=\sum_{{\rm J}^{[m]\backslash \{1\}}}u_{j_2}^{a_2}\dots u_{j_m}^{a_m}\sum_{\substack{j_1\neq j_{m+1}\\ j_1,j_{m+1}\neq j_2,\dots,j_{m}}}\sum_{j=1}^{2}(-1)^j \sum_{k=0}^{\lfloor \frac{a_1-j}{2}\rfloor}u_{j_1}^ku_{j_{m+1}}^{a_1-j-k}-\frac{1}{2}e(a_1-j)u_{j_1}^{\frac{a_1-j}{2}}u_{j_{m+1}}^{\frac{a_1-j}{2}}\\
            &=\sum_{j=1}^2\sum_{k=0}^{\lfloor \frac{a_1-j}{2}\rfloor}\sum_{{\rm J}^{[m]}}(-1)^ju_{j_2}^{a_2}\dots u_{j_m}^{a_m}\bigg(u_{j_1}^ku_{j_{m+1}}^{a_1-j-k}-\frac{1}{2}e(a_1-j)u_{j_1}^{\frac{a_1-j}{2}}u_{j_{m+1}}^{\frac{a_1-j}{2}}\bigg).
        \end{align*}
        Similarly, the second sum equals
            \begin{multline*}
            \sum_{k=2}^{m}{\rm sgn}(a_1-a_k)\sum_{j=1}^{2} (-1)^{j} \\ \times \sum_{l={\rm min}\{a_1+1-j,a_k\}}^{\lfloor \frac{a_1+a_k-j}{2}\rfloor}\sum_{{\rm J}^{[m]}}u_{j_2}^{a_2}\dots u_{j_m}^{a_m}\bigg(u_{j_1}^lu_{j_{k}}^{a_1+a_k-j-l}-\frac{1}{2}e(a_1+a_k-j)u_{j_1}^{\frac{a_1+a_k-j}{2}}u_{j_{m+1}}^{\frac{a_1+a_k-j}{2}}\bigg).
        \end{multline*}
        Therefore, after substituting these into \eqref{p1}, equation \eqref{recur-rl-I} is proved. The equations for the initial conditions can be seen from the formulas
        $$
            \sum_{{\rm J}^{[m]}}\prod_{i=1}^{m-1} u_{j_i}^{a_{i+1}}=(\beta-m+1)\sum_{{\rm J}^{[m-1]}}\prod_{i=1}^{m-1} u_{j_i}^{a_{i+1}},\quad
            \bigg(\sum_{j=1}^{\beta}u_j\bigg)^m=\sum_{\tau\in\mathcal{P}(m)}\frac{m!}{\tau!s_\tau!}\sum_{{\rm J}^{[|\tau|]}}\prod_{k=1}^{|\tau|}u_{j_k}^{\tau_k}.
        $$
    \end{proof}

    \begin{remark}
        From the expansion of $\prod_{j=1}^\beta\left(1-\left(1-e^{ i\theta / N}\right) u_j\right)^{N-2}$, one knows that the multiple integral in \eqref{4.12b} can be written as a linear combination of $\mathcal{I}^{(m)}(a_1,a_2,\dots,a_m) $, which upon use of Proposition \ref{propA1} can be expressed in terms of $\theta$-derivatives of the leading term $\mathcal{I}[1]$. Therefore, all the higher-order correction terms of the two-point correlation can be expressed as a finite linear combination of the derivatives of the leading term with coefficients being polynomials in $\theta$. However, to make this explicit beyond the $N^{-2}$ term, in particular for the $N^{-4}$ term, would be an arduous task due to the many terms involved and
        multiple cancellations.
        In the case $\beta =2$, it was noted in (\ref{R1z}) that a simple differential relation linking
        $\rho_{(2),2,\beta=2}^{\rm bulk}(s,0) $ and $\rho_{(2),0,\beta=2}^{\rm bulk}(s,0)$ can be read off from the respective exact functional forms.
    \end{remark}

         \appendix
\section*{Appendix B}
\renewcommand{\thesection}{B} 
\setcounter{equation}{0}
\setcounter{prop}{0}

It turns out that the leading  small-$x$ 
(small-$s$ after introducing scaled variables)
form of the coefficient of $\xi^k$ in (\ref{4.2c}) can be determined, along with the $N$-dependent factor multiplying this term. According to  (\ref{4.2d}) and (\ref{1.15}), for $\beta = 2$, 
and with bulk scaling as on the LHS of (\ref{4.3b}) (scaled variable $s$)
we read off that for $k=0$ and $k=1$ respectively, the  leading  small-$s$ coefficients of $\xi^k$ are
\begin{equation}\label{c1}
{(1-1/N^2) \pi^2 s^2 \over 3}, \quad
- {(1-4/N^2)(1-1/N^2)^2 \pi^6 s^7 \over 4050}.
\end{equation}
The significance of knowledge of this is that it provides necessary conditions for the differential identity (\ref{5.1b}), as each such coefficient of $\xi^k$ must independently have their  $N$ independent part related to the $N^{-2}$ dependent part according to (\ref{5.1b}), and similarly should there be such relations for parts at order $N^{-4}$ etc.

\begin{prop}
Consider (\ref{4.2c}) for the circular $\beta$ ensemble with bulk scaled variables as on the LHS of (\ref{4.3b}).
The  leading  small-$s$ coefficient of $\xi^k$ is, for $n:=k+2 < N/2$ and $\kappa := \beta/2$ a positive integer, equal to
\begin{multline}\label{c1B}
{(-s)^k \over k!} \Big ( {2 \pi s\over N} \Big )^{ \kappa (k+2)(k+1)} 
 S_k(\beta,\beta,\kappa) \\
 \times
{(\kappa!)^n \over \Gamma(n \kappa +1)} 
\bigg ( \prod_{j=1}^{n-1} {\Gamma(\kappa j + 1) \over \Gamma(\kappa(n+j) + 1)} \bigg )
\bigg (  \prod_{j=1}^{n-1}  \prod_{l=1}^{j \kappa} (\kappa^2 N^2 - l^2) \bigg )
\bigg |_{n=k+2}.
\end{multline}

\end{prop}

\begin{proof}
Analogous to the structure of (\ref{4.5}) in the limit of small $\{r_j\}$, when to leading order all the $r_j$ dependent terms on the second line can be set equal to unity, in this regime we have for the general $\beta > 0$ bulk scaled $n$-point correlation
\begin{multline}\label{4.5d}
\tilde{\rho}_{(n),\beta}\left(r_1, \ldots, r_n\right)\sim \\ \frac{(N-n+1)_n((\beta/2)!)^N}{N^n\Gamma\left(\frac{\beta N}{2}+1\right)}M_{N-n}(n \beta / 2, n \beta / 2, \beta / 2)\prod_{1\leq j<k\leq n}
\Big |{2\pi r_k \over N}-{2\pi r_j \over N} \Big |^\beta.       
\end{multline}
The relevance of this to (\ref{4.2c}) with bulk scaling as on the LHS of (\ref{4.3b}) is seen by changing variables in the multiple integrals therein $x_j \mapsto 2 \pi s x_j/N$.
The integrand can then be substituted by (\ref{4.5d}),
showing that to leading order for $s$-small, the coefficient of $\xi^k$ is given by
\begin{multline}\label{c1A}
{(-s)^k \over k!} \Big ( {2 \pi s\over N} \Big )^{ \beta (k+2)(k+1)/2} 
 S_k(\beta,\beta,\beta/2) \\
 \times
\frac{(N-n+1)_n((\beta/2)!)^N}{N^n\Gamma\left(\frac{\beta N}{2}+1\right)}
      M_{N-n}(n \beta / 2, n \beta / 2, \beta / 2) \Big |_{n=k+2}.
\end{multline}
Here $S_k$ refers to the Selberg integral.
According to the working of the proof of Proposition \ref{P3.1}, the second line simplifies to the second line of (\ref{c1A}), and thus the result.
\end{proof}

The $N$-dependent factors in (\ref{c1A}), namely $N^{-\kappa n (n-1)} $ times the double product on the second line, with $n=k+2$, have the large $N$ form
\begin{equation}\label{R1w}
\kappa^{n(n-1)} \bigg ( 1 - {1 \over N^2}v_2(\kappa) + {1 \over 2 N^4} \Big (
(v_2(\kappa))^2 - v_4(\kappa) \Big ) + 
{\rm O}\Big ( {1 \over N^6} \Big ) \bigg ),
\end{equation}
where $v_s(\kappa) := \kappa^{-2}\sum_{j=1}^{n-1} \sum_{l=1}^{j \kappa} l^s$. In particular
$$
v_2(\kappa) = { n \over 12 \kappa }(n-1) (\kappa(n-1) + 1) (\kappa n + 1).
$$
Since ${1 \over 6 \beta} {d^2 \over d s^2} s^{n + \kappa n (n - 1)} = v_2(\kappa)
s^{(n-2) + \kappa n ( n - 1)}$, where the term being differentiated on the LHS is the $s$-dependent factor in (\ref{c1B}) times $s^2$, we see that there is consistency with the conjecture (\ref{5.1Ba}).
The coefficient of $N^{-4}$ in (\ref{R1w}) is an order eight polynomial in $n$, telling us that if there were to be an analogue of (\ref{5.1Ba}) at this order, the highest power derivative must be four.

\small
\providecommand{\bysame}{\leavevmode\hbox to3em{\hrulefill}\thinspace}
\providecommand{\MR}{\relax\ifhmode\unskip\space\fi MR }
\providecommand{\MRhref}[2]{%
  \href{http://www.ams.org/mathscinet-getitem?mr=#1}{#2}
}
\providecommand{\href}[2]{#2}

  \end{document}